\pdfoutput=1
\documentclass[sigconf]{acmart}

\usepackage{algorithm}
\usepackage[noend]{algpseudocode}
\algnewcommand{\algorithmicassumption}{\textbf{Requirement:}}
\algnewcommand{\Assume}{\item[\algorithmicassumption]}
\algrenewcomment[1]{\hfill \(\triangleright\) \emph{\small #1}} % to italicize text in comments
\algnewcommand{\InlineIf}[2]{% single line if-then
  \algorithmicif\ #1\ \algorithmicthen\ #2}
\algnewcommand{\InlineElse}[1]{% single line else
  \algorithmicelse\ #1}
\algnewcommand{\InlineIfElse}[3]{% single line if-then-else
  \algorithmicif\ #1\ \algorithmicthen\ #2\ \algorithmicelse\ #3}
\algnewcommand{\InlineFor}[2]{\algorithmicfor\ #1\ \algorithmicdo\ #2} % single line for loop
\algnewcommand{\CommentLine}[1]{\(\triangleright\) \emph{\small #1}}

\algnewcommand{\algorithmicand}{\textbf{and}}
\algnewcommand{\algorithmicor}{\textbf{or}}
\algnewcommand{\FOR}{\algorithmicfor}
\algnewcommand{\OR}{\algorithmicor}
\algnewcommand{\AND}{\algorithmicand}
\algnewcommand{\IF}{\algorithmicif}
\algnewcommand{\THEN}{\algorithmicthen}
\algnewcommand{\ELSE}{\algorithmicelse}
\usepackage[frozencache,cachedir=minted-issac22]{minted}
\usepackage[shortlabels]{enumitem}
\usepackage[capitalise]{cleveref}
\newcommand{\bigO}[1]{\mathchoice{O\left(#1\right)}{O(#1)}{O(#1)}{O(#1)}} % big O for complexity
\newcommand{\softO}[1]{\mathchoice{\tilde{O}\left(#1\right)}{O\tilde{~}(#1)}{O\tilde{~}(#1)}{O\tilde{~}(#1)}} % soft O for complexity
\newcommand{\expmm}{\omega} % exponent for the cost of matrix multiplication
\newcommand{\ZZ}{\mathbb{Z}} % relative integers
\newcommand{\NN}{\mathbb{Z}_{\ge 0}} %  integers
\newcommand{\ZZp}{\mathbb{Z}_{> 0}} % positive integers
\newcommand{\card}[1]{\##1}  % cardinality of a set
\newcommand{\var}{x} % default variable for univariate polynomials
\newcommand{\field}{\mathbb{K}} % base field
\newcommand{\polRing}{\field[\var]} % polynomial ring
\newcommand{\matRing}[2]{\field^{#1 \times #2}} % matrix ring over base field
\newcommand{\pmatRing}[2]{\polRing^{#1 \times #2}} % univariate polynomial matrix ring
\newcommand{\rdim}{m} % default row dimension
\newcommand{\cdim}{n} % default column dimension
\newcommand{\rk}{r} % generic symbol for rank
\newcommand{\rp}{j} % generic symbol for rank profile
\newcommand{\rps}{J} % generic symbol for rank profile tuple
\newcommand{\stind}{\theta} % starting index for rank profile tuple in ZLS algo
\newcommand{\row}[1]{\mathbf{#1}} % for a row of a matrix
\newcommand{\col}[1]{\mathbf{#1}} % for a column of a matrix
\newcommand{\mat}[1]{\mathbf{#1}} % for a matrix
\newcommand{\matrow}[2]{{#1}_{#2,*}} % for the row #2 of the matrix #1
\newcommand{\matrows}[2]{{#1}_{#2,*}} % for the submatrix of #1 of its rows in #2
\newcommand{\matcol}[2]{{#1}_{*,#2}} % for the column #2 of the matrix #1
\newcommand{\matcols}[2]{{#1}_{*,#2}} % for the submatrix of #1 of its columns in #2
\newcommand{\matsub}[3]{{#1}_{#2,#3}} % for the submatrix of #1 of its rows in #2 and columns in #3
\newcommand{\trsp}[1]{#1^\mathsf{T}} %transpose
\newcommand{\diag}[1]{\mathrm{diag}(#1)}  % diagonal matrix with diagonal entries #1
\newcommand{\xdiag}[1]{\mat{\var}^{#1\,}} % shift matrix, diagonal of powers of X with exponents given by #1
\newcommand{\idMat}[1]{\mat{I}_{#1}} % identity matrix of size mxm
\newcommand{\matz}{\mat{0}}  % zero matrix
\newcommand{\anyMat}{\boldsymbol{\ast}}  % indicate block in a matrix : bold asterisk
\newcommand{\rank}[1]{\mathrm{rank}(#1)}
\newcommand{\tuple}[1]{\boldsymbol{#1}}  % tuples (mainly used for integers I think)
\newcommand{\sumTuple}[1]{|#1|} % sum of entries in a tuple
\newcommand{\subTuple}[2]{{#1}_{#2}} % subtuple
\newcommand{\rdeg}[2][]{\mathrm{rdeg}_{{#1}}(#2)} % shifted row degree
\newcommand{\cdeg}[2][]{\mathrm{cdeg}_{{#1}}(#2)} % shifted column degree
\newcommand{\lmat}[2][]{\mathrm{lm}_{#1}(#2)} % leading matrix of polynomial matrix, default shifts = 0 (uniform)
\newcommand{\shiftz}{\mathbf{0}} % notation for uniform shifts ~ (0,..,0)
\newcommand{\shifts}{\tuple{s}} % shifts vector
\newcommand{\shiftt}{\tuple{t}} % shifts vector, bis
\newcommand{\shiftr}{\tuple{\hat{s}}} % shifts vector, modified (reduced entries)
\newcommand{\cdegs}{\tuple{d}} % tuple of column degrees
\newcommand{\pivDeg}{\delta} % entry of the minimal degree or pivot degree
\newcommand{\pivDegs}{\boldsymbol{\pivDeg}} % minimal degree
\newcommand{\pivInd}{\pi} % entry of the pivot index
\newcommand{\pivInds}{\boldsymbol{\pi}} % pivot index
\newcommand{\nonPivInds}{\boldsymbol{\pi}^c} % index of entries not in pivot
\newcommand{\dd}{D} % sum of shift, degree of determinant
\newcommand{\modRel}[2]{\operatorname{\mathcal{R}}_{#1}(#2)} % general relation module
\newcommand{\modKer}[1]{\operatorname{\mathcal{K}}(#1)} % kernel module
\newcommand{\order}{\tau} % order for approximant basis
\newcommand{\pt}{\alpha} % letter for points for interpolant bases
\newcommand{\sys}{\mat{F}} % input matrix for kernel
\newcommand{\kbas}{\mat{K}} % matrix for kernel basis
\newcommand{\kvec}{\row{p}} % vector in kernel
\newcommand{\abas}{\mat{A}} % matrix for approximant basis
\newcommand{\mmat}{\mat{M}} % modulus matrix, for relations mod M
\newcommand{\pmat}{\mat{P}} % notation for ``general'' polynomial matrix
\copyrightyear{2022}
\acmYear{2022}
\setcopyright{acmlicensed}\acmConference[ISSAC '22]{Proceedings of the 2022 International Symposium on Symbolic and Algebraic Computation}{July 4--7, 2022}{Villeneuve-d'Ascq, France}
\acmBooktitle{Proceedings of the 2022 International Symposium on Symbolic and Algebraic Computation (ISSAC '22), July 4--7, 2022, Villeneuve-d'Ascq, France}
\acmPrice{15.00}
\acmDOI{10.1145/3476446.3535495}
\acmISBN{978-1-4503-8688-3/22/07}

\begin{document}
\fancyhead{}

\title[Rank-Sensitive Computation of the Rank Profile of a Polynomial Matrix]
{Rank-Sensitive Computation of the Rank Profile \texorpdfstring{\\}{} of a Polynomial Matrix}

\author{George Labahn}
\affiliation{%
  \institution{Cheriton School of Computer Science}%
  \city{University of Waterloo}%
  \country{Ontario, Canada} }
  %\email{glabahn@uwaterloo.ca}

\author{Vincent Neiger}
\affiliation{%
	\institution{Sorbonne Universit\'e, CNRS, LIP6}
	\city{F-75005 Paris}
	\postcode{75252}\country{France}}
%\email{vincent.neiger@lip6.fr}

\author{Thi Xuan Vu}
\affiliation{%
  \institution{Department of Mathematics and Statistics}
  \city{UiT, The Arctic University of Norway, Troms\o{}} 
  \postcode{N-9037}\country{Norway}
}
%\email{thi.x.vu@uit.no}

\author{Wei Zhou}
\affiliation{%
  \institution{Cheriton School of Computer Science}%
  \city{University of Waterloo}%
  \country{Ontario, Canada} }
%\email{wei2learn@gmail.com}

\begin{abstract}
Consider a matrix $\mathbf{F} \in \mathbb{K}[x]^{m \times n}$ of univariate
polynomials over a field~$\mathbb{K}$. We study the problem of computing the
column rank profile of $\mathbf{F}$. To this end we first give an algorithm
which improves the minimal kernel basis algorithm of Zhou, Labahn, and
Storjohann (Proceedings ISSAC 2012). We then provide a second algorithm which
computes the column rank profile of $\mathbf{F}$ with a rank-sensitive
complexity of $O\tilde{~}(r^{\omega-2} n (m+D))$ operations in $\mathbb{K}$.
Here, $D$ is the sum of row degrees of $\mathbf{F}$, $\omega$ is the exponent of
matrix multiplication, and $O\tilde{~}(\cdot)$ hides logarithmic factors. 
\end{abstract}

\maketitle

\section{Introduction}
\label{sec:intro}

In this paper, we consider the computation of rank properties of a univariate
polynomial matrix \(\sys \in \pmatRing{\rdim}{\cdim}\) over some base field
\(\field\). The rank of \(\sys\) can be determined by computing a basis for the
left (or for the right) kernel of \(\sys\). Under the assumption $\rdim \ge
\cdim$ (which implicitly requires the input matrix to have full rank, see
\cref{sec:kbas_rkprof}), an algorithm due to Zhou, Labahn, and Storjohann
\cite{ZhLaSt12} computes a minimal basis for the left kernel of $\sys$ using
$\softO{\rdim^{\expmm} \lceil \rho/\cdim \rceil}$ operations in \(\field\),
where \(\rho\) is the sum of the \(\cdim\) largest row degrees of \(\sys\).
%% When \(\rdim\) is not much larger than \(\cdim\), \(\rho/\cdim\) is close to
%% the average row degree of \(\sys\).
In this cost bound, $\expmm$ is the exponent of matrix multiplication, and
\(\softO{\cdot}\) is $\bigO{\cdot}$ but ignoring logarithmic factors. A natural
alternative is to compute a basis for the row space of $\sys$, called a row
basis (or, similarly, a column basis).  However, the fastest known row basis
algorithm \cite{ZhoLab13} starts by computing a basis of the left kernel of
\(\sys\), so one may as well get the rank directly from the latter.

Currently the best known cost bound for computing the rank of \(\sys\) only
depends on the matrix dimensions $\rdim$ and $\cdim$, and is not influenced by
the rank $\rk$. More generally, the fastest known algorithms for basic
computations with univariate polynomial matrices are not rank-sensitive. This
is a significant drawback for the manipulation of matrices whose rank is
unknown, and possibly low a priori.

Furthermore, there are specific situations where rank deficiency is actually
expected by design, and one would like to take advantage of this in algorithms.
Recently, in a context of computing generators of linearly recurrent sequences,
minimal approximant bases of rank-deficient, structured matrices \(\sys\) have
been encountered \cite[Sec.\,5]{HyunNeigerSchost2021}. It has also been
observed that, for the computation of the Hermite normal form of \(\sys\),
finding the \emph{(column) rank profile} of \(\sys\) provides a direct
reduction to the case of a square, nonsingular matrix \cite{NeiRosSol18}, for
which fast methods are known \cite{LaNeZh17}. A third situation occurs in
verification protocols for polynomial matrices: most protocols proposed in
\cite{LucasNeigerPernetRocheRosenkilde2021} rely, directly or indirectly, on
one for certifying ``\(\rank{\sys} \ge \gamma\)''
\cite[Prot.\,3]{LucasNeigerPernetRocheRosenkilde2021}, itself asking the Prover
to locate a square, nonsingular submatrix of \(\sys\) which has rank at least
\(\gamma\).

In fast \(\field\)-linear algebra, rank-sensitive algorithms and complexity
bounds have proved highly valuable. For example, rank-sensitive Gaussian
elimination costs \(\bigO{\rk^{\expmm-2} \rdim \cdim}\) operations in
\(\field\) \cite{StorjohannMulders1998,Storjohann00,Jeannerod2006}, for an
input (constant) matrix \(\mat{A} \in \matRing{\rdim}{\cdim}\) of rank \(\rk\).
More recently, research on this topic has led to improvements of both
complexity bounds and software implementations, and has also provided deep
insight into the rank-related properties that are revealed depending on the
chosen elimination strategies
\cite{JeannerodPernetStorjohann2013,DumasPernetSultan2017}. For finding the
rank or rank profile and for solving linear systems,
\cite{CheungKwokLau2013,StorjohannYang2015} report on running times as low as
\((\rk^\expmm + \rdim + \cdim + |\mat{A}|)^{1+o(1)}\), with \(|\mat{A}|\) the
number of nonzero entries of \(\mat{A}\). Now, for univariate polynomial
matrices, despite the impact this would have on many computations, there is
still an important lack of efficient rank-sensitive methods which would
incorporate both fast linear algebra techniques and fast univariate polynomial
multiplication.

One possibility is to make use of classical algorithms such as fraction-free
Gaussian elimination (see \cite{GeddesCzaporLabahn92}) while also keeping track
of row or column operations to obtain a kernel basis and rank profile. The cost
of such algorithms depends on \(\rdim\), \(\cdim\), \(\rk\) and the degree of
matrices but does not involve the exponent of matrix multiplication \(\expmm\).
This is also the case for the algorithm of Mulders and Storjohann
\cite{MulSto03} which transforms \(\sys\) to weak Popov form and computes the
rank profile with a cost of \(\bigO{\rk \rdim \cdim  \deg(\sys)^2}\). Storjohann
\cite[Chap.\,2]{Storjohann00} gives a recursive version of fraction-free
Gaussian elimination which computes a kernel basis and rank profile of \(\sys\)
having complexity of \(\softO{\rk^{\expmm-1} \rdim \cdim \deg(\sys)^2}\)
operations in \(\field\). Storjohann and Villard \cite{StoVil05} later gave a
Las Vegas randomized algorithm which computed the rank and kernel basis of a
polynomial matrix with complexity \(\softO{\rk^{\expmm-2} \rdim \cdim
\deg(\sys)}\).

The main contribution of this paper is a column rank profile algorithm with a
rank-sensitive cost of \(\softO{\rk^{\expmm-2} \cdim (\rdim + \dd)}\)
operations in \(\field\). Here \(\dd\) is the sum of the row degrees of
\(\sys\), with \(\dd \le \rdim\deg(\sys)\). This is a follow-up and improvement
to the algorithm given in the PhD thesis of Zhou \cite[Sec.\,11]{Zhou12}.

We first revisit \cite[Algo.\,11.1]{Zhou12}, to augment the minimal kernel
basis algorithm of \cite{ZhLaSt12} so that it also determines the column rank
profile of the input matrix. How the variant here improves upon those in the
last two references is explained at the beginning of \cref{sec:kbas_rkprof}. In
particular, within the same complexity bound, the new version supports
arbitrary dimensions \(\rdim, \cdim\) and rank \(\rk\) of \(\sys\), which is
essential for our purpose. This algorithm is not rank-sensitive: it has a cost
of \(\softO{\rdim^{\expmm-2} (\rdim+\cdim) (\rdim+\dd)}\) operations in
\(\field\).

We then give a rank-sensitive column rank profile algorithm, which uses the
above kernel basis algorithm as its main subroutine. A sketch of a similar
result has been given before in \cite[Sec.\,11.2]{Zhou12}, where the approach
is to incorporate the \emph{columns} gradually, always maintaining a number of
columns which is bounded by the rank. At each step the above kernel basis
procedure is called to obtain a partial column rank profile and discard rows
that are \(\polRing\)-linearly dependent. At each step as well, to keep control
of the cost of this kernel computation, a row basis computation is applied
beforehand to reduce to a matrix having full row rank.

Here, we follow another path, by incorporating \emph{rows} gradually. This
allows us to benefit from the fact that the kernel procedure has quasi-linear
cost with respect to the column dimension \(\cdim\), without having to resort
to row basis computations. To enable proceeding row-wise, we exploit a property
of kernel bases in so-called weak Popov form, showing that they give direct
access to a set of linearly independent rows of the input in addition to its
column rank profile. Once all rows of \(\sys\) have been processed and a set of
\(\rk\) linearly independent rows of \(\sys\) has been found, the column rank
profile of \(\sys\) can be extracted efficiently again through the kernel
algorithm.

\emph{Outline.} In \cref{sec:preliminaries}, we give the basic definitions and
properties of our building blocks for polynomial matrix arithmetic including
kernel bases, pivot profiles and rank profiles, and weak Popov forms.
\cref{sec:rank_properties} introduces specific rank profile and kernel
properties used in our algorithms.  \cref{sec:kbas_rkprof} describes our
algorithm for computing the rank profile and kernel basis, while
\cref{sec:rkprof_indeprows} presents our algorithm for the rank-sensitive
computation of the rank profile. The paper ends with topics for future
research, and contains as an appendix an illustration of our examples through
SageMath code.

\section{Preliminaries}
\label{sec:preliminaries}

In this section we describe the notations used in this paper, and then give the
basic definitions and a number of properties of polynomial matrices including
\emph{shifted degrees} and \emph{pivot profiles}, \emph{relation bases} and
{\em kernel bases}, \emph{reduced forms} and \emph{weak Popov forms}.

\subsection{Notation}

We let $\polRing$ denote a univariate polynomial
ring over a field $\field$ with  $\pmatRing{\rdim}{\cdim}$ being the
set of ${\rdim} \times {\cdim}$ univariate polynomial matrices.
For $\sys \in \pmatRing{\rdim}{\cdim}$ and subsets $I$ of $(1, \ldots, \rdim)$
and $J$ of $(1, \dots, \cdim)$, we write \(\matsub{\sys}{I}{J}\) for the
submatrix of $\sys$ obtained by selecting rows indexed by $I$ and columns
indexed by $J$. We let  $\matrows{\sys}{I} = \matsub{\sys}{I}{\{1..n\}}$ denote
the submatrix of $\sys$ obtained by selecting the rows indexed by $I$ and
keeping all columns and $\matcols{\sys}{J} = \matsub{\sys}{\{1..m\}}{J}$ for
the submatrix of $\sys$ obtained by keeping all rows and selecting columns
indexed by $J$.

For a tuple of integers \(\shifts = (s_1,\ldots,s_\rdim) \in\ZZ^\rdim\), the
sum of its entries is denoted by \(\sumTuple{\shifts} = s_1+\cdots+s_\rdim\).
When this concerns an input shift \(\shifts\), we will often write \(\dd\) for
this quantity, i.e.~\(\dd = \sumTuple{\shifts}\).

\subsection{Kernel, row space, modules of relations}
\label{sec:preliminaries:notation}

For a matrix $\sys$ in $\pmatRing{\rdim}{\cdim}$ of rank $r$, the set 
\[
  \modKer{\sys} := \left\{ \row{p} \in \pmatRing{1}{\rdim} \mid \row{p} \sys = \matz \right\}
\]
is a $\polRing$-module of rank $\rdim-\rk$ and is called the (left)
\emph{kernel} of $\sys$. The \emph{row space} of \(\sys\) is the module
\[
  \{\row{p} \sys \mid \row{p} \in \pmatRing{1}{\rdim}\} \subseteq \pmatRing{1}{\cdim}.
\]
A basis for one of these modules (a kernel basis or a row basis) is typically
organized into a single polynomial matrix, for example, a basis of
\(\modKer{\sys}\) being represented by a full rank matrix \(\mat{K} \in
\pmatRing{(\rdim-\rk)}{\rdim}\). Also, for a nonsingular matrix $\mmat$ in
$\pmatRing{\cdim}{\cdim}$,
\[
  \modRel{\mmat}{\sys} :=  \left\{ \row{p} \in \pmatRing{1}{\rdim} \mid \row{p} \sys = \row{0} \bmod \mmat \right\}
\]
is a $\polRing$-module of rank $m$, called the (left) \emph{relation module}
of $\sys$ modulo $\mmat$. Here, the notation $\mat{A} = \matz \bmod \mmat$
means that $\mat{A} = \mat{Q}\mmat$ for some matrix $\mat{Q}$. 

Important particular cases are the relations of \emph{approximation} and those
of \emph{interpolation} \cite{Beckermann92,BarBul92,BecLab94}. For the latter,
\(\mmat = \diag{M_1,\ldots,M_\cdim}\) with \(M_k = (\var - \pt_{k,1}) \cdots
(\var-\pt_{k,\order_k})\) for some \(\order_k \in \ZZp\) and \(1 \le k \le
\cdim\), where the \(\pt_{k,j}\)'s are known elements from \(\field\).
Approximation is when these elements are zero: \(\mmat =
\diag{\var^{\order_1},\ldots,\var^{\order_\cdim}}\), so that working mod
\(\mmat\) amounts to truncating the column \(j\) modulo \(\var^{\order_j}\).

The notions of right kernel, column space, column bases, and right relations
are of course defined similarly.

\subsection{Shifted degrees, leading matrix}
\label{sec:preliminaries:sdegrees}

For a row vector \(\row{p} = [p_1 \;\cdots\; p_\cdim]\) in
\(\pmatRing{1}{\cdim}\), its \emph{degree} is \[\rdeg{\row{p}} = \max_{1 \le i \le \cdim} \deg(p_i)\]
that is, the largest degree of all its entries. Here we take the
convention that the degree of a zero polynomial or zero row is \(-\infty\). In many cases it is useful to shift (or re-weigh) the degrees.
Given a shift \(\shifts = (s_1, \dots, s_\cdim) \in
\ZZ^\cdim\), the \emph{\(\shifts\)-degree} of \(\row{p}\) is defined as
\[\rdeg[\shifts]{\row{p}} = \max_{1 \le i \le \cdim} (s_i+\deg(p_i)).\] Note that this is
equal to \(\rdeg[]{\row{p}\xdiag{\shifts}}\), where \(\xdiag{\shifts}\) is the
diagonal matrix with diagonal entries \(x^{s_1}, \ldots, x^{s_\cdim}\).

For a matrix $\pmat \in \pmatRing{\rdim}{\cdim}$ and a shift $\shifts \in
\ZZ^\cdim$, the row degree $\rdeg[]{\pmat}$ of $\pmat$ is the list of the
degrees of its rows, and similarly the \emph{$\shifts$-row degree}
$\rdeg[\shifts]{\pmat}$ is the list of the \(\shifts\)-degrees of its rows.
Then, the \emph{$\shifts$-leading matrix} $\lmat[\shifts]{\pmat}$ of $\pmat$ is
the matrix in \(\matRing{\rdim}{\cdim}\) formed by the coefficients of degree
zero of $\xdiag{-\shiftt}\pmat \xdiag{\shifts}$, where $\shiftt =
\rdeg[\shifts]{\pmat}$. By convention, a zero row in \(\pmat\) yields a zero
row in \(\lmat[\shifts]{\pmat}\). 

\subsection{Pivot and rank profiles}
\label{sec:preliminaries:hnf_rrp}

If the row vector \(\row{p}\) is nonzero, the \emph{$\shifts$-pivot index} of
$\row{p}$ is the largest index $\pivInd$ such that $\deg(p_\pivInd) +
s_\pivInd = \rdeg[\shifts]{\row{p}}$. In this case $p_\pivInd$ and
$\deg(p_\pivInd)$ are the \emph{$\shifts$-pivot entry} and the
\emph{$\shifts$-pivot degree} of $\row{p}$. Note that \(\pivInd\) is also the
index of the rightmost nonzero entry in \(\lmat[\shifts]{\row{p}}\).

The pair \((\pivInds,\pivDegs) = (\pivInd_i,\pivDeg_i)_{1\le i\le\rdim}\) where
\(\pivInds = (\pivInd_i)_{1\le i\le\rdim}\) and \(\pivDegs = (\pivDeg_i)_{1\le
i\le\rdim}\) are the \(\shifts\)-pivot index and degree for each row of the
matrix \(\mat{P}\), is called the \emph{\(\shifts\)-pivot profile} of $\pmat$.
Observe that \(\rdeg[\shifts]{\mat{P}}\) is equal to \((\pivDeg_{i} +
s_{\pivInd_i})_{1\le i\le\rdim}\).

The \emph{(column) rank profile} of \(\mat{P}\) is the lexicographically
minimal list of integers \(J = (j_1,\ldots,j_\rk)\) such that
\(\matcols{\mat{P}}{J}\) has rank \(\rk = \rank{\sys}\). In what follows, rank
profile always means \emph{column} rank profile; otherwise, we will write
explicitly \emph{row} rank profile.

A matrix \(\mat{H} = [h_{i,j}] \in \pmatRing{\rk}{\cdim}\) with \(\rk \le
\cdim\) is in \emph{Hermite normal form}
\cite{Hermite1851,MacDuffee33,Newman72} if there are indices $1 \le j_1 <
\cdots < j_\rk \le \cdim$ such that
\begin{itemize}
  \item[(i)] for $1\le i\le\rk$, $h_{i,j_i} \neq 0$ is monic and $h_{i,j} = 0$ for
    $1 \le j < j_i$;
  \item[(ii)] for $1\le k < i\le\rk$, $\deg(h_{k,j_i})<\deg(h_{i,j_i})$.
\end{itemize}
In this case, $(j_1,\ldots,j_\rk)$ is the rank profile of $\mat{H}$.

The Hermite normal form of $\mat{P} \in \pmatRing{\rdim}{\cdim}$ is its unique
row basis \(\mat{H} \in \pmatRing{\rk}{\cdim}\), with \(\rk=\rank{\mat{P}}\),
which is in Hermite normal form. Then, the rank profile of \(\mat{P}\)
is equal to that of \(\mat{H}\), since \(\mat{U}\mat{P} = [\begin{smallmatrix}
\mat{H} \\ \matz \end{smallmatrix}]\) for some unimodular matrix \(\mat{U} \in
\pmatRing{\rdim}{\rdim}\).

\subsection{Reduced forms, predictable degree}
\label{sec:preliminaries:reduced_forms}

With the above definitions, \(\mat{P}\) is said to be \emph{$\shifts$-row
reduced} if $\lmat[\shifts]{\pmat}$ has full row rank. A core feature of these
matrices is the \emph{predictable degree property}, which says that there
cannot be any cancellation of high-degree terms of the matrix via
\(\polRing\)-linear combinations of the rows (see \citep{Forney75,Kailath80}
for the case \(\shifts=\shiftz\); \citep[Lem.\,3.6]{BeLaVi99} for any
\(\shifts\); and \citep[Thm.\,1.11]{Neiger16b} for a proof of the equivalence
in the next lemma).

\begin{lemma}[Predictable degree]
  \label{lem:predictable_degree}
  Let \(\pmat \in \pmatRing{\rdim}{\cdim}\) have no zero row, let
  \(\shifts\in\ZZ^\cdim\) and \(\shiftt = \rdeg[\shifts]{\pmat}\). Then,
  \(\pmat\) is in \(\shifts\)-reduced form if and only if
  \(\rdeg[\shifts]{\mat{Q}\pmat} = \rdeg[\shiftt]{\mat{Q}}\) for all \(\mat{Q}
  \in \pmatRing{k}{\rdim}\).
\end{lemma}

\begin{corollary}
  \label{cor:leading_matrix_product}
  Let \(\pmat \in \pmatRing{\rdim}{\cdim}\) and let \(\shifts\in\ZZ^\cdim\) be
  such that \(\pmat\) is \(\shifts\)-reduced. Let \(\shiftt =
  \rdeg[\shifts]{\pmat}\). Then, \(\lmat[\shifts]{\mat{Q}\pmat} =
  \lmat[\shiftt]{\mat{Q}}\lmat[\shifts]{\pmat}\) for any \(\mat{Q} \in
  \pmatRing{k}{\rdim}\).
\end{corollary}
\begin{proof}
  Let \(\tuple{d} = \rdeg[\shiftt]{\mat{Q}} \in \ZZ^k\). By the predictable
  degree property, \(\tuple{d} = \rdeg[\shifts]{\mat{Q}\pmat}\). The
  conclusion then follows from the identity
  \[
    \xdiag{-\tuple{d}} \mat{Q} \pmat \xdiag{\shifts}
    =
    (\xdiag{-\tuple{d}} \mat{Q} \xdiag{\shiftt}) (\xdiag{-\shiftt} \pmat \xdiag{\shifts}).
    \qedhere
  \]
\end{proof}

As a consequence, shifted reduced forms are preserved by multiplication,
provided the shifts are appropriately chosen. This result is at the core of
divide and conquer algorithms for bases of relation modules
\citep{BecLab94,GiJeVi03} and kernel bases \citep[Thm.\,3.9]{ZhLaSt12}.

\begin{lemma}
  \label{lem:reduced_product}
  Let \(\pmat \in \pmatRing{\rdim}{\cdim}\) and \(\shifts \in \ZZ^\cdim\) such
  that \(\pmat\) is in \(\shifts\)-reduced form. Let \(\shiftt =
  \rdeg[\shifts]{\pmat} \in \ZZ^\rdim\), and let \(\mat{Q} \in
  \pmatRing{k}{\rdim}\) be in \(\shiftt\)-reduced form. Then, \(\mat{Q} \pmat\)
  is in \(\shifts\)-reduced form.
\end{lemma}
\begin{proof}
  The assumptions imply that
  both \(\lmat[\shiftt]{\mat{Q}}\) and \(\lmat[\shifts]{\pmat}\) have full row
  rank.  Thus their product has full row rank as well, and according to
  \cref{cor:leading_matrix_product} this product is
  \(\lmat[\shifts]{\mat{Q}\pmat}\).
\end{proof}

\subsection{Weak Popov forms, predictable pivot}
\label{sec:preliminaries:weak_popov}

A matrix $\pmat = [p_{i,j}] \in \pmatRing{\rk}{\cdim}$ with no zero row is
\emph{$\shifts$-Popov} if
\begin{itemize}
  \item[(i)] its $\shifts$-pivot index \((\pivInd_1,\ldots,\pivInd_\rk)\) is strictly
    increasing;
  \item[(ii)] for \(1\le i \le \rk\), \(p_{i,\pivInd_i}\) is monic;
  \item[(iii)] for each \(k , i \in \{1,\ldots,\rk\}\) with \(k\neq i\),
    \(\deg(p_{k,\pivInd_i}) < \deg(p_{i,\pivInd_i})\).
\end{itemize}

\noindent
If $\pmat$ only satisfies the first condition, it is said to be \emph{$\shifts$-weak
Popov}. Any \(\shifts\)-weak Popov matrix is \(\shifts\)-reduced.
Furthermore, each matrix has a unique row basis in \(\shifts\)-Popov form.

We remark that, for weak Popov forms, it is sometimes only required (see
e.g.~\cite{MulSto03}) that the pivot indices be pairwise distinct, instead of
increasing. Then, the forms with the added requirement of increasing indices
were called ordered weak Popov forms. Here, we will only manipulate ordered
weak Popov forms, and therefore we call them weak Popov forms for ease of
presentation.

The shifted weak Popov form satisfies the following refinement of the
predictable degree property and is also compatible with multiplication under
well-chosen shifts (see \citep[Sec.\,5]{BeLaVi06} for related considerations
and \cite[Lem.\,2.6]{NeigerPernet2021} for a proof of the next lemmas).

\begin{lemma}[Predictable pivot]
  \label{lem:predictable_pivot}
  Let \(\pmat \in \pmatRing{\rdim}{\cdim}\) have no zero row, let
  \(\shifts\in\ZZ^\cdim\) and \(\shiftt = \rdeg[\shifts]{\pmat}\), and let
  \((\pivInd_i,\pivDeg_i)_{1\le i\le\rdim}\) be the \(\shifts\)-pivot profile
  of \(\pmat\). If \(\pmat\) is in \(\shifts\)-weak Popov form, then the
  \(\shifts\)-pivot profile of \(\mat{Q}\pmat\) is
  \((\pivInd_{j_i},\pivDeg_{j_i}+d_i)_{1\le i\le k}\) for all \(\mat{Q} \in
  \pmatRing{k}{\rdim}\), where \((j_i,d_i)_{1\le i\le k}\) is the
  \(\shiftt\)-pivot profile of \(\mat{Q}\).
\end{lemma}

\begin{lemma}
  \label{lem:weak_popov_multiplication}
  Let \(\pmat \in \pmatRing{\rdim}{\cdim}\) and \(\shifts \in \ZZ^\cdim\) such
  that \(\pmat\) is in \(\shifts\)-weak Popov form. Let \(\shiftt =
  \rdeg[\shifts]{\pmat} \in \ZZ^\rdim\), and let \(\mat{Q} \in
  \pmatRing{k}{\rdim}\) be in \(\shiftt\)-weak Popov form. Then, \(\mat{Q}
  \pmat\) is in \(\shifts\)-weak Popov form.
\end{lemma}
\begin{proof}
  By assumption, the \(\shifts\)-pivot index \((\pivInd_i)_{1\le i\le\rdim}\)
  of \(\pmat\) and the \(\shiftt\)-pivot index \((j_i)_{1\le i\le k}\) of
  \(\mat{Q}\) are both strictly increasing. Then, by
  \cref{lem:predictable_pivot}, the \(\shifts\)-pivot index of \(\mat{Q}\pmat\)
  is the subtuple \((\pivInd_{j_i})_{1\le i\le k}\), which is strictly
  increasing.  Hence \(\mat{Q}\pmat\) is in \(\shifts\)-weak Popov form.
\end{proof}

Note however that a similar product of shifted Popov forms does not yield a
shifted Popov form, but only a shifted \emph{weak} Popov form.

\subsection{Example}

We will use the following as a running example in this paper.

\begin{example}
  \label{example:main}
  Working over \(\field=\mathbb{F}_2\), let $\sys \in \pmatRing{5}{5}$ be given by
\[
  \begin{bmatrix}
    {x}^{2}&{x}^{3}+1&{x}^{8}+{x}^{6}+{x}^{4
    }+{x}^{3}+{x}^{2}+x&{x}^{4}+1&{x}^{3}+1\\ 
    0&{x}^{4}+ 1&{x}^{5}+{x}^{4}+{x}^{3}+{x}^{2}&x+1&{x}^{2}+1\\
    0& {x}^{2}+1&x+1&0&1\\ 0&0&{x}^{8}+1&{x}^{4}+1&0 \\ 
    0&0&{x}^{4}+1&1&0
  \end{bmatrix}
\]
Then the matrix 
\[
  \begin{bmatrix}
      0&1&{x}^{2}+1&0&x+1\\ 
    0 &1&{x}^{2}+1&1&{x}^{4}+x
  \end{bmatrix}
  \in \pmatRing{2}{5}
\]
is a weak Popov basis of $\modKer{\sys}$ (which is not in Popov form). It has
pivot index $\pivInds = (3, 5)$ and pivot degree $\pivDegs =
 (2,4)$. Here is now an
\(\shifts\)-weak Popov basis of  \(\modKer{\sys}\) for the shift
\(\shifts = \rdeg{\sys} = (8,5,2,8,4)\): 
\[
  \kbas =
  \begin{bmatrix}
    0  & x^3 & x^5 + x^3 & 1 & x^3 + 1 \\
    0  &   1 &   x^2 + 1 & 0 &   x + 1
  \end{bmatrix}
  \in \pmatRing{2}{5}.
\]
Its \(\shifts\)-pivot index is \((4,5)\) and its $\shifts$-pivot degree is
\((0, 1)\).
\qed
\end{example}

The above kernel bases were computed using the SageMath software, as described
in \cref{fig:sage_code} on \cpageref{fig:sage_code}.

\section{Rank and degree properties related to kernel bases}
\label{sec:rank_properties}

In this section we discuss rank and degree properties related to kernel bases
and which are central for the correctness and complexity of the algorithms in
\cref{sec:kbas_rkprof,sec:rkprof_indeprows}.

\begin{lemma}
  \label{lem:ker_submatrix}
  Let \(\sys \in \pmatRing{\rdim}{\cdim}\) have rank \(\rk\). For any \(\mat{V}
  \in \pmatRing{\cdim}{k}\) such that \(\sys\mat{V}\) has rank \(\rk\), we have
  \(\modKer{\sys} = \modKer{\sys\mat{V}}\). As a corollary, if \(\rps \subseteq
  \{1,\ldots,\cdim\}\) is such that \(\matcol{\sys}{\rps}\) has rank \(\rk\),
  then \(\modKer{\sys} = \modKer{\matcol{\sys}{\rps}}\).
\end{lemma}
\begin{proof}
  The second statement follows from the first one, by building \(\mat{V}\) from
  the columns of the identity matrix \(\idMat{\cdim}\) with index in \(\rps\).

  Concerning the first statement, the rank assumption implies that the left
  kernels \(\modKer{\sys}\) and \(\modKer{\sys\mat{V}}\) have the same rank
  \(\rdim-\rk\). Let \(\kbas_1 \in \pmatRing{(\rdim-\rk)}{\rdim}\) and
  \(\kbas_2 \in \pmatRing{(\rdim-\rk)}{\rdim}\) be bases of \(\modKer{\sys}\)
  and \(\modKer{\sys\mat{V}}\), respectively. It is clear that \(\modKer{\sys}
  \subseteq \modKer{\sys\mat{V}}\), hence \(\kbas_1 = \mat{U}\kbas_2\) for some
  nonsingular \(\mat{U} \in \pmatRing{(\rdim-\rk)}{(\rdim-\rk)}\). The fact
  that kernel bases have unimodular column bases \cite{ZhoLab14} ensures that
  \(\mat{U}\) is unimodular, and thus \(\modKer{\sys} =
  \modKer{\sys\mat{V}}\).
\end{proof}

\begin{lemma}
  \label{lem:ker_crp_equal}
  If \(\sys \in \pmatRing{\rdim}{\cdim}\) and \(\mat{G} \in
  \pmatRing{\ell}{\cdim}\) are two matrices which have the same right kernel,
  then \(\sys\) and \(\mat{G}\) have the same column rank profile. As a
  corollary, if \(I \subseteq \{1,\ldots,\rdim\}\) is such that
  \(\matrow{\sys}{I}\) has the same rank as \(\sys\), then \(\matrow{\sys}{I}\)
  has the same rank profile as \(\sys\).
\end{lemma}
\begin{proof}
  The second statement follows from the first: applying
  \cref{lem:ker_submatrix} to \(\trsp{\sys}\) and \(\trsp{(\matrow{\sys}{I})} =
  \matcol{(\trsp{\sys})}{I}\) shows that these matrices have the same left
  kernel, i.e.~\(\sys\) and \(\matrow{\sys}{I}\) have the same right kernel.

  Let \(\rps_1\) and \(\rps_2\) be the rank profiles of \(\sys\) and
  \(\mat{G}\), respectively. Let \(j \in \{1,\ldots,\cdim\}\) be such that \(j
  \not\in \rps_1\), meaning that there exists a vector \(\col{u} = \trsp{[u_1
  \cdots u_j]} \in \pmatRing{j}{1}\) such that \(u_j \neq 0\) and
  \(\matcol{\sys}{1..j} \, \col{u} = 0\). Since the right kernel of \(\sys\) is
  contained in that of \(\mat{G}\), it follows that \(\matcol{\mat{G}}{1..j}
  \col{u} = 0\), and therefore \(j\not\in\rps_2\). We have proved \(\rps_2
  \subseteq \rps_1\), and the same arguments prove \(\rps_1 \subseteq \rps_2\),
  by symmetry.  Hence \(\rps_2 = \rps_1\).
\end{proof}

\begin{theorem}
  \label{thm:degree_rank}
  Let \(\sys \in \pmatRing{\rdim}{\cdim}\) have rank \(\rk\), let \(\shifts \in
  \ZZ^\rdim\), and let \(\kbas \in \pmatRing{(\rdim-\rk)}{\rdim}\) be an
  \(\shifts\)-weak Popov basis of \(\modKer{\sys}\). Let
  \((\pivInds,\pivDegs)\) be the \(\shifts\)-pivot profile of \(\kbas\), and
  let \(\nonPivInds = \{1,\ldots,\rdim\}\setminus\pivInds\) be the indices of
  the columns of \(\kbas\) which do not contain an \(\shifts\)-pivot entry of
  \(\kbas\). Assume also that \(\sys\) factors as \(\sys = \mat{S} \mat{R}\)
  where \(\mat{R} \in \pmatRing{\rk}{\cdim}\) and \(\mat{S} \in
  \pmatRing{\rdim}{\rk}\). Then,
  \begin{enumerate}[(i)]
    \item \label{thm:degree_rank:rank}
      \(\matrows{\sys}{\nonPivInds} \in \pmatRing{\rk}{\cdim}\)
      has rank \(\rk\), which is the size of \(\nonPivInds\);
    \item \label{thm:degree_rank:pivdeg}
      \(\matrows{\mat{S}}{\nonPivInds} \in \pmatRing{\rk}{\rk}\) is nonsingular
      and \(\sumTuple{\pivDegs} \le
      \deg(\det(\matrows{\mat{S}}{\nonPivInds}))\), hence in particular
      \(\sumTuple{\pivDegs} \le \sumTuple{\rdeg{\matrows{\sys}{\nonPivInds}}}
      \le \rk\deg(\sys)\);
    \item \label{thm:degree_rank:rdeg}
      if \(\shifts\ge \rdeg{\sys}\), then \(\sumTuple{\rdeg[\shifts]{\kbas}}
      \le \sumTuple{\shifts}\).
  \end{enumerate}
\end{theorem}

Concerning the matrices \(\mat{R}\) and \(\mat{S}\), note that they have rank
\(\rk\), since otherwise we would have \(\rank{\sys} = \rank{\mat{S}\mat{R}} <
\rk\). Taking a row basis of \(\sys\) for \(\mat{R}\) proves the existence of
such matrices.

\cref{thm:degree_rank:rank} states that, from any shifted weak Popov basis of
the left kernel of \(\sys\), we can immediately deduce a set of
\(\rk=\rank{\sys}\) rows of \(\sys\) which are \(\polRing\)-linearly
independent.  \cref{thm:degree_rank:rdeg} is the main degree property that was
exploited in the design of the fastest known minimal kernel basis algorithm
\cite{ZhLaSt12} (see \cite[Thm.\,3.4]{ZhLaSt12}), explaining also why this
algorithm restricts to shifts such that \(\shifts\ge \rdeg{\sys}\). Here we
prove it as a consequence of the property in \cref{thm:degree_rank:pivdeg},
which gives more precise degree information in particular through better
accounting for the rank of \(\sys\).

We now prove \cref{thm:degree_rank}.

\begin{proof}
  For \cref{thm:degree_rank:rank} it suffices to prove that
  \(\modKer{\matrows{\sys}{\nonPivInds}} = \{\matz\}\).  Let \(\row{v} \in
  \pmatRing{1}{\rk}\) be such that \(\row{v} \matrows{\sys}{\nonPivInds} =
  \matz\).  Construct \(\row{w} \in \pmatRing{1}{\rdim}\) such that
  \(\matcols{\row{w}}{\pivInds} = \matz\) and \(\matcols{\row{w}}{\nonPivInds}
  = \row{v}\). The vector \(\row{w}\) is in \(\modKer{\sys}\), so  \(\row{w} =
  \row{u} \kbas\) for some \(\row{u} \in \pmatRing{1}{(\rdim-\rk)}\), and hence
  \(\matz = \matcols{\row{w}}{\pivInds} = \row{u} \matcols{\kbas}{\pivInds}\).
  Thus  \(\row{u}=\matz\) since \(\matcols{\kbas}{\pivInds}\) is nonsingular,
  implying \(\row{v} = \matcols{\row{w}}{\nonPivInds} = \row{u}
  \matcols{\kbas}{\nonPivInds} = \matz\).

  To prove \cref{thm:degree_rank:pivdeg} set
  \[
    \mat{S}_1=\matrows{\mat{S}}{\nonPivInds} \in \pmatRing{\rk}{\rk}
    \text{ and }
    \mat{S}_2=\matrows{\mat{S}}{\pivInds} \in \pmatRing{(\rdim-\rk)}{\rk}
  \]
  as well as
  \[
    \kbas_1=\matcols{\kbas}{\nonPivInds} \in \pmatRing{(\rdim-\rk)}{\rk}
    \text{ and }
    \kbas_2=\matcols{\kbas}{\pivInds} \in \pmatRing{(\rdim-\rk)}{(\rdim-\rk)}.
   \]
   Then $\mat{S}_1$ is nonsingular since \(\matrows{\sys}{\nonPivInds} = \mat{S}_1
  \mat{R}\) has rank \(\rk\).
 We are going to prove that \(\kbas_2\)
  is an \(\subTuple{\shifts}{\pivInds}\)-weak Popov basis of
  \(\modRel{\mat{S}_1}{\mat{S}_2}\), where \(\subTuple{\shifts}{\pivInds} \in
  \ZZ^{\rdim-\rk}\) is the subshift of \(\shifts\) formed by its entries with
  index in \(\pivInds\). From this, \citep[Cor.\,2.4]{NeiVu17} ensures that
  \(\deg(\det(\kbas_2)) = \sumTuple{\pivDegs} \le \deg(\det(\mat{S}_1))\).

  Since \(\sys = \mat{S}\mat{R}\), with \(\mat{R}\) full row rank, we
  have \(\modKer{\sys} = \modKer{\mat{S}}\) and so \(\kbas\) is an
  \(\shifts\)-weak Popov basis of \(\modKer{\mat{S}}\). From \(\kbas \mat{S} =
  \matz\) we obtain \(\kbas_1 \mat{S}_1 + \kbas_2 \mat{S}_2 = \matz\) and so
  the rows of \(\kbas_2\) are in \(\modRel{\mat{S}_1}{\mat{S}_2}\).  It remains
  to show that any \(\kvec \in \modRel{\mat{S}_1}{\mat{S}_2}\) is a
  \(\polRing\)-linear combination of the rows of \(\kbas_2\). By definition of
  \(\modRel{\mat{S}_1}{\mat{S}_2}\), there exists \(\row{q} \in
  \pmatRing{1}{\rk}\) such that \(\kvec \mat{S}_2 = \row{q} \mat{S}_1\).
  Considering \(\row{v} \in \pmatRing{1}{\rdim}\) such that
  \(\matcols{\row{v}}{\pivInds} = \kvec\) and \(\matcols{\row{v}}{\nonPivInds}
  = -\row{q}\), we have \(\row{v} \mat{S} =
  \matcols{\row{v}}{\nonPivInds}\mat{S}_1 + \matcols{\row{v}}{\pivInds}
  \mat{S}_2 = -\row{q} \mat{S}_1 + \kvec \mat{S}_2 = \matz\), that is,
  \(\row{v} \in \modKer{\mat{S}}\). Thus, \(\row{v} = \row{u} \kbas\) for some
  \(\row{u} \in \pmatRing{1}{(\rdim-\rk)}\), and we obtain \(\kvec =
  \matcols{\row{v}}{\pivInds} = \row{u} \matcols{\kbas}{\pivInds} = \row{u}
  \kbas_2\).

  Let \(\sys_1 = \matrows{\sys}{\nonPivInds} \in \pmatRing{\rk}{\cdim}\) and
  \(\shiftt = \rdeg{\sys_1} \in \ZZ^{\rk}\). In order to prove the last two
  bounds on \(\sumTuple{\pivDegs}\), observe that \(\sumTuple{\shiftt} \le
  \rk\deg(\sys)\) is clear since \(\sys_1\) consists of \(\rk\) rows of
  \(\sys\). It remains to show that \(\deg(\det(\mat{S}_1)) \le
  \sumTuple{\shiftt}\). Let \(\mat{U} \in \pmatRing{\cdim}{\rk}\) be such that
  \(\sys_1 \mat{U}\) is the \(-\shiftt\)-column Popov form of \(\sys_1\).
  Since \(\cdeg[-\shiftt]{\sys_1} \le \shiftz\), the minimality of the
  shifted column degrees of shifted reduced forms \cite[Sec.\,2.7]{Zhou12}
  implies \(\cdeg[-\shiftt]{\sys_1 \mat{U}} \le \shiftz\) as well. According to
  \cite[Lem.\,2.2]{ZhoLab13}, this translates as \(\rdeg{\sys_1\mat{U}} \le
  \shiftt\), and so \(\sumTuple{\rdeg{\sys_1\mat{U}}} \le
  \sumTuple{\shiftt}\). Since \(\sys_1 \mat{U}\) is \(-\shiftt\)-column Popov,
  it is also row reduced, and therefore \(\sumTuple{\rdeg{\sys_1\mat{U}}} =
  \deg(\det(\sys_1\mat{U}))\) \cite[Sec.\,6.3.2]{Kailath80}. It follows that
  \(\deg(\det(\sys_1\mat{U})) \le \sumTuple{\shiftt}\) and, using \(\sys_1 =
  \mat{S}_1 \mat{R}\), we obtain
  \[
    \deg(\det(\mat{S}_1)) + \deg(\det(\mat{R}\mat{U}))
    = \deg(\det(\mat{S}_1 \mat{R}\mat{U}))
    \le \sumTuple{\shiftt}.
  \]

  To prove \cref{thm:degree_rank:rdeg}, recall
  that \(\rdeg[\shifts]{\kbas} = (\pivDeg_{i} + s_{\pivInd_i})_{1\le
  i\le\rdim-\rk}\), and therefore \(\sumTuple{\rdeg[\shifts]{\kbas}} =
  \sumTuple{\pivDegs} + \sumTuple{\shifts_{\pivInds}}\). From
  \cref{thm:degree_rank:pivdeg} we get that \(\sumTuple{\rdeg[\shifts]{\kbas}}
  \le \sumTuple{\rdeg{\matrows{\sys}{\nonPivInds}}} +
  \sumTuple{\shifts_{\pivInds}}\), and from the assumption \(\shifts \ge
  \rdeg{\sys}\) we conclude that \(\sumTuple{\rdeg[\shifts]{\kbas}} \le
  \sumTuple{\shifts_{\nonPivInds}} + \sumTuple{\shifts_{\pivInds}} =
  \sumTuple{\shifts}\).
\end{proof}

\begin{example}
  \label{example:cont_main}
  Following on from \cref{example:main}, consider the matrices $\sys$ and
  $\kbas$ and the shift $\shifts = \rdeg{\sys} = (8,5,2,8,4)$. Since the
  $\shifts$-pivot index of $\kbas$ is $\pivInds = (4, 5)$, the indices of the
  columns of $\kbas$ which do not contain an \(\shifts\)-pivot entry are
  $\pivInds^c = (1,2,3)$.   

  Regarding \cref{thm:degree_rank:rank}, from the above \(\shifts\)-pivot
  information we get that the rank of $\sys$ is $3$ and the rows $(1,2,3)$ of
  $\sys$ are $\polRing$-linearly independent. The other kernel basis considered
  in \cref{example:main} shows that the rows \((1,2,4)\) are also
  \(\polRing\)-linearly independent.

  Regarding \cref{thm:degree_rank:pivdeg}, observe that the row degree of
  $\matrow{\sys}{\nonPivInds}$ is \((8, 5, 2)\), so
  $\sumTuple{\rdeg[\shifts]{\matrow{\sys}{\nonPivInds}}} = 15$. From
  \cref{example:main}, the $\shifts$-pivot degree of \(\mat{K}\) is \(\pivDegs
  = (8, 7)\), so $\sumTuple{\pivDegs} = 15$. Furthermore here \(\rk \deg(\sys)
  = 3 \cdot 8 = 24\). Thus, here we have \(\sumTuple{\pivDegs} =
  \sumTuple{\rdeg{\matrows{\sys}{\nonPivInds}}} \le \rk\deg(\sys)\).

  Finally, regarding \cref{thm:degree_rank:rdeg},
  $\sumTuple{\rdeg[\shifts]{\kbas}} = \sumTuple{(8,5)} = 13$, which is bounded
  from above by $\sumTuple{\shifts} = 27$.  \qed
\end{example}
 
\section{Computing the rank profile and a kernel basis}
\label{sec:kbas_rkprof}

In this section we give an improved version of the minimal kernel
basis algorithm in \cite{ZhLaSt12}. In addition to the new algorithm
we also  include a proof of correctness and determine its complexity. 

\subsection{Algorithm}
\label{sec:kbas_rkprof:algo}

Our improvements of the algorithm, compared to the versions in in
\cite{ZhLaSt12} \cite[Sec.\,11]{Zhou12}, is summarized as follows:

\begin{itemize}
\item[(i)]
Besides a kernel basis, the algorithm also finds the \emph{column rank profile}
of \(\sys\), without additional operations, based on the approach in
\cite[Sec.\,11.1]{Zhou12}.
\item[(ii)]
The output kernel basis \(\kbas\) is in \emph{\(\shifts\)-weak Popov form}
instead of \(\shifts\)-reduced form. This has the advantage of revealing the
\(\shifts\)-pivot profile, which can be used for example to further transform
\(\kbas\) into \(\shifts\)-Popov form \cite[Sec.\,5]{NeigerPernet2021}. Thanks
to \cref{thm:degree_rank:rank} of \cref{thm:degree_rank}, this also reveals a
set of \(\rank{\sys}\) rows of \(\sys\) that are \(\polRing\)-linearly
independent, a property that we exploit in \cref{algo:column_rank_profile}.
\item[(iii)]
The algorithm \emph{supports any input matrix \(\sys\)}, without assumption on
its rank or dimensions. In comparison, the assumption \(\rdim\ge\cdim\) is made
in the complexity analysis in \cite{ZhLaSt12,Zhou12}, which implicitly requires
that the input \(\sys\) have full column rank (indeed, if \(\sys\) is
rank-deficient, the algorithm in these references cannot guarantee that the
assumption \(\rdim\ge\cdim\) is satisfied in recursive calls).
\item[(iv)]
The algorithm may use \emph{any relation basis}
(\crefrange{step:algo:kbas_rkprof:rbas:start}{step:algo:kbas_rkprof:rbas:end}),
instead of restricting to approximant bases. As early experiments have showed
\cite[Sec.\,4.2]{HyuNeiSch19}, this can lead to speed-ups at least by constant
factors, for example by relying on well-chosen interpolation bases. Still, as
seen in \cref{thm:algo:kbas_rkprof,sec:kbas_rkprof:complexity}, for one
specific point of the complexity analysis we restrict to relation bases modulo
a diagonal matrix.
\end{itemize}

\begin{theorem}
  \label{thm:algo:kbas_rkprof}
  Let \(\sys\in\pmatRing{\rdim}{\cdim}\) have rank \(\rk\), and let
  \(\shifts\in\NN^\rdim\) such that \(\shifts \ge \rdeg{\sys}\). The call
  \(\Call{KernelBasis-RankProfile}{\sys,\shifts}\) returns an \(\shifts\)-weak
  Popov basis \(\kbas \in \pmatRing{(\rdim-\rk)}{\rdim}\) of \(\modKer{\sys}\)
  and the column rank profile \((\rp_1,\ldots,\rp_\rk) \in \ZZp^\rk\) of
  \(\sys\). Assuming that \(\rdim \in \bigO{\cdim}\) and that one chooses a
  matrix \(\mmat\) at \cref{step:algo:kbas_rkprof:rbas:chooseM} which is diagonal with
  all entries of degree \(\tau\), this algorithm uses $\softO{\rdim^{\expmm-2}
  (\rdim+\cdim) (\rdim+\dd)}$ operations in \(\field\), where \(\dd =
  \sumTuple{\shifts}\).
\end{theorem}

\begin{algorithm}[ht]
  \caption{\textsc{KernelBasis-RankProfile}\((\sys,\shifts)\)}
  \label{algo:kbas_rkprof}
  \begin{algorithmic}[1]
    \Require{a matrix \(\sys\in\pmatRing{\rdim}{\cdim}\), a shift \(\shifts\in\NN^\rdim\)}
    \Assume{\(\shifts \ge \rdeg{\sys}\) entrywise}
    \Ensure{an \(\shifts\)-ordered weak Popov basis \(\kbas \in
      \pmatRing{(\rdim-\rk)}{\rdim}\) of \(\modKer{\sys}\) and the column rank
      profile \((\rp_1,\ldots,\rp_\rk) \in \ZZp^\rk\) of \(\sys\)}

    \If{\(\sys = \matz\)} \Comment{kernel of zero is identity} \label{step:algo:kbas_rkprof:zero:start}
      \State\Return \(\idMat{\rdim} \in \pmatRing{\rdim}{\rdim}\), \(() \in \ZZp^0\) \label{step:algo:kbas_rkprof:zero:end}
    \EndIf

    \If{\(\rdim=1\)} \Comment{kernel of nonzero \(1\times\cdim\) matrix is empty} \label{step:algo:kbas_rkprof:onerow:start}
      \State \(\rp \in \{1,\ldots,\cdim\} \gets\) index of first nonzero entry in \(\sys\)
      \State \Return \([] \in \pmatRing{0}{1}\), \((\rp) \in \ZZp^1\) \label{step:algo:kbas_rkprof:onerow:end}
    \EndIf

    \If{\(\rdim < 2\cdim\)}\Comment{``wide'' matrix: divide and conquer on columns}
        \State \(\sys_1 \in \pmatRing{\rdim}{\lfloor \cdim/2 \rfloor} \gets \matcols{\sys}{\{1,\ldots,\lfloor\cdim/2\rfloor\}}\)
              \label{step:algo:kbas_rkprof:split_first}
        \State \(\kbas_1 \in \pmatRing{\ell_1}{\rdim},\rps_1 \in \ZZp^{\rk_1} \gets\)
              \label{step:algo:kbas_rkprof:split_rec_one}
        \Statex \hfill \(\Call{KernelBasis-RankProfile}{\sys_1,\shifts}\) 
        \State \(\sys_2 \in \pmatRing{\ell_1}{\lceil \cdim/2 \rceil} \gets \kbas_1 \cdot \matcols{\sys}{\{\lfloor\cdim/2\rfloor+1,\ldots,\cdim\}}\)
               \label{step:algo:kbas_rkprof:split_residual}
        \State \(\kbas_2 \in \pmatRing{\ell_2}{\ell_1},\rps_2 \in \ZZp^{\rk_2} \gets\)
               \label{step:algo:kbas_rkprof:split_rec_two}
        \Statex \hfill \(\Call{KernelBasis-RankProfile}{\sys_2,\rdeg[\shifts]{\kbas_1}}\)
        \State\CommentLine{note:
            \(\rk_1 = \rank{\sys_1}\),
            \(\rk_2 = \rank{\sys_2}\),
            \(\ell_1 = \rdim-\rk_1\),
            and \(\ell_2 = \ell_1 - \rk_2\)}
        \State \(\rps_2 \in \ZZp^{\rk_2} \gets\) shift \(\rps_2\) by adding \(\lfloor n/2 \rfloor\) to all entries
        \State \Return \(\kbas_2 \cdot \kbas_1, (\rps_1,\rps_2) \in \ZZp^{\rk_1+\rk_2}\)
            \label{step:algo:kbas_rkprof:split_return}
    \EndIf

    \State \CommentLine{from here we are in the case \(\sys\neq\mat{0}\), \(\rdim\ge2\), and \(\cdim\le\frac{m}{2}\)}
    \State \CommentLine{minimize shift while preserving \(\shiftr \ge \rdeg{\sys}\)}
            \label{step:algo:kbas_rkprof:rbas_params:start}
    \State \(\mu \gets \min(\shifts-\rdeg{\sys})\);
            \(\shiftr \gets \shifts - (\mu,\ldots,\mu)\);
            \(\order \gets \left\lceil \frac{2 \sumTuple{\shiftr}}{\rdim-\cdim} \right\rceil\)
            \label{step:algo:kbas_rkprof:rbas_params:end}
    \State \CommentLine{choose type of relations and compute relation basis}
           \label{step:algo:kbas_rkprof:rbas:start}
    \State \(\mmat \in \pmatRing{\cdim}{\cdim}\gets\) choose any
            matrix in Hermite normal form with \(\min(\cdeg{\mmat}) \ge \order\)
            \Comment{for example \(\mmat = \diag{\var^\tau,\ldots,\var^\tau}\)}
           \label{step:algo:kbas_rkprof:rbas:chooseM}
    \State \(\abas \gets\) \(\shiftr\)-weak Popov basis of \(\modRel{\mmat}{\sys}\)
           \label{step:algo:kbas_rkprof:rbas:end}
    \State \CommentLine{compute residual and indices \(I\) of rows already in kernel}
         \label{step:algo:kbas_rkprof:residual:start}
    \State \(I \gets \{i \in \{1,\ldots,\rdim\} \mid \rdeg[\shiftr]{\matrow{\abas}{i}} < \order \}\) \Comment{rows in kernel}
    \State \(I^c \gets \{1,\ldots,\rdim\} \setminus I\) \Comment{rows expected not to be in kernel}
    \State \(\mat{G} \gets \matrows{\abas}{I^c} \, \sys \, \mmat^{-1}\)
         \label{step:algo:kbas_rkprof:residual:product}
    \Comment{if \(\mmat = \diag{\var^\tau,\ldots,\var^\tau}\), this is \(\var^{-\order}\matrows{\abas}{I^c} \, \sys\)}
    \State if \(\mat{G}\) has zero rows, update \((I,I^c,\mat{G})\) accordingly
         \label{step:algo:kbas_rkprof:residual:end}
    \State \CommentLine{compute kernel of residual recursively and merge results}
            \label{step:algo:kbas_rkprof:reccall:start}
    \State \(\shiftt \gets \rdeg[\shiftr]{\matrows{\abas}{I^c}}-(\gamma,\ldots,\gamma)\) where \(\gamma = \min(\cdeg{\mmat})\)
            \label{step:algo:kbas_rkprof:reccall:shift}
    \State \(\kbas_2,\rps \gets \Call{KernelBasis-RankProfile}{\mat{G},\shiftt}\)
            \label{step:algo:kbas_rkprof:reccall:reccall}
    \State \(\kbas \gets\) matrix formed by both the rows of \(\kbas_2\cdot
    \matrows{\abas}{I^c}\) and those of \(\matrows{\abas}{I}\), sorted by
    increasing \(\shiftr\)-pivot index
            \label{step:algo:kbas_rkprof:reccall:product}
    \State \Return \(\kbas\), \(\rps\)
          \label{step:algo:kbas_rkprof:rbas_return}
  \end{algorithmic}
\end{algorithm}

\begin{example}
  \label{example:continued}
  Let $\sys \in \mathbb{F}_2[x]^{5 \times 5}$ be the matrix from
  \cref{example:main}, and consider the shift \(\shifts=\rdeg{\sys} =
  (8,5,2,8,4)\). At the top level of the recursion, \cref{algo:kbas_rkprof}
  first finds the kernel basis of the $5 \times 2$ submatrix $\sys_1 =
  \matcol{\sys}{1..2}$, via a recursive call. With \(5 \ge 2\cdot 2\), this
  call runs
  \crefrange{step:algo:kbas_rkprof:split_first}{step:algo:kbas_rkprof:split_return},
  with \(\order = \lceil \frac{2\cdot 27}{3} \rceil = 18\). This eventually
  yields
  \[
    \kbas_1 =
    \begin{bmatrix}
      0 & 1 & x^2 + 1 & 0 & 0 \\
      0 & 0 & 0 & 1 & 0 \\
      0 & 0 & 0 & 0 & 1
    \end{bmatrix}
    \text{ and rank profile }
    \rps_1 = (1,2).
  \]
  In this case, using an approximant basis of $\sys_1$ at order \(\order\)
  yields the three above rows of the kernel, and two additional rows (this can
  be observed by running the code in \cref{fig:sage_code} on
  \cpageref{fig:sage_code}). On this specific example it is easily observed
  that \(\rank{\sys_1} = 2\), so one may infer that \(\kbas_1\) and \(\rps_1\)
  are directly deduced from \(\abas\), without running the recursive call at
  \cref{step:algo:kbas_rkprof:reccall:reccall}.

  Multiplying $\kbas_1$ by the last three columns of $\sys$ gives
  \[
  \sys_2 =
  \begin{bmatrix}
    x^5+x^4+x+1 & x+1 & 0\\
    x^8+1 & x^4+1 & 0\\
    x^4+1 & 1 & 0
  \end{bmatrix}
  \in \pmatRing{3}{3}.
  \]
  Since this matrix has \(\rdim<2\cdim\), we then recurse along the first
  column of $\sys_2$, with shift \(\rdeg[\shifts]{\kbas_1} = (5,8,4)\). This
  gives a \((5,8,4)\)-weak Popov basis of the kernel of that column, as
  \[
    \kbas_1' =
    \begin{bmatrix}
      x^3 & 1 & x^3 + 1 \\
        1 & 0 &   x + 1
    \end{bmatrix}
    \in \pmatRing{2}{3}.
  \]
  This also provides the rank profile \((1)\) of that column.

  Multiplying \(\kbas_1'\) by the last two columns of $\sys_2$ gives a zero
  matrix, with the identity as the kernel basis. Hence $\kbas_2 = \kbas_1'$ is
  the sought \((5,8,4)\)-weak Popov basis of $\modKer{\sys_2}$, and the rank
  profile of \(\sys_2\) is \(\rps_2 = (1)\). Then the latter is shifted to
  \(\rps_2 = (1+\lfloor 5/2 \rfloor) = (3)\), to keep track of the position of
  the column block \(\sys_2\) in the input \(\sys\).

  Concatenating \(\rps_1\) and \(\rps_2\) yields the rank profile \((1,2,3)\) of
  \(\sys\), and the product $\kbas_2 \kbas_1$ is the kernel basis \(\kbas\)
  given in \cref{example:main}.
  \qed
\end{example}

\subsection{Proof of correctness}
\label{sec:kbas_rkprof:correctness}

In this subsection, we prove the correctness of \cref{algo:kbas_rkprof}.

\textbf{Cases \(\sys=\matz\) or \(\rdim=1\).} The correctness of
\crefrange{step:algo:kbas_rkprof:zero:start}{step:algo:kbas_rkprof:onerow:end}
is clear.

\textbf{Case \(2\cdim>\rdim\).} Here the algorithm runs
\crefrange{step:algo:kbas_rkprof:split_first}{step:algo:kbas_rkprof:split_return}
and returns. We assume correctness for the recursive calls at
\cref{step:algo:kbas_rkprof:split_rec_one,step:algo:kbas_rkprof:split_rec_two}.
From \(\shifts\ge\rdeg{\sys}\), we get \(\rdeg[\shifts]{\kbas_1} \ge
\rdeg{\kbas_1\sys} \ge \rdeg{\sys_2}\) and the requirement of the call at
\cref{step:algo:kbas_rkprof:split_rec_two} is satisfied.
\cref{lem:weak_popov_multiplication} implies that the matrix \(\kbas_2
\kbas_1\) is in \(\shifts\)-weak Popov form.  Furthermore,
\[
  \kbas_2 \kbas_1 \sys
  = \kbas_2 \kbas_1 [\sys_1 \;\; \matcols{\sys}{\{\lfloor\cdim/2\rfloor+1,\ldots,\cdim\}}]
  = \kbas_2 [\matz \;\; \sys_2]
  = \matz.
\]
To prove that \(\kbas_2\kbas_1\) generates the kernel \(\modKer{\sys}\), we let \(\row{p}
\in \modKer{\sys}\) and prove that \(\row{p} = \row{u} \kbas_2
\kbas_1\) for some \(\row{u} \in \pmatRing{1}{\ell_2}\). Since \(\row{p}\) is
in \(\modKer{\sys_1}\), and \(\kbas_1\) is a basis of the latter
kernel, we have that \(\row{p} = \row{v} \kbas_1\) for some \(\row{v} \in
\pmatRing{1}{\ell_1}\). By construction of \(\sys_2\), \(\row{p} \sys = [\matz
\;\; \row{v} \sys_2]\).  Then \(\row{p}\sys = \matz\) implies \(\row{v} \in
\modKer{\sys_2}\) and,  since \(\kbas_2\) is a basis of the latter kernel, we
have \(\row{v} = \row{u} \kbas_2\) for some \(\row{u} \in
\pmatRing{1}{\ell_2}\). This yields \(\row{p} = \row{u}
\kbas_2\kbas_1\). Thus  \(\kbas_2\kbas_1\) is an \(\shifts\)-weak
Popov basis of \(\modKer{\sys}\).

In order to prove that \((\rps_1,\rps_2)\) is the rank profile of \(\sys\), the
main observation is that since \(\kbas_1\) has full row rank, it can be
completed into a nonsingular matrix \(\mat{U} = [\begin{smallmatrix} \anyMat \\
\kbas_1 \end{smallmatrix}] \in \pmatRing{\rdim}{\rdim}\). Then,
\[
  \mat{U} \sys
  =
  %% \begin{bmatrix}
  %%   \anyMat \\
  %%   \kbas_1
  %% \end{bmatrix}
  %% \begin{bmatrix}
  %%   \sys_1 \;\; \matcols{\sys}{\{\lfloor\cdim/2\rfloor+1,\ldots,\cdim\}}
  %% \end{bmatrix}
  %% =
  \begin{bmatrix}
    \mat{V} & \anyMat \\
    \matz & \sys_2
  \end{bmatrix}
  \text{ for some } \mat{V} \in \pmatRing{\rk_1}{\lfloor \cdim/2 \rfloor}.
\]
\(\mat{V}\) has the same rank profile as \([\begin{smallmatrix} \mat{V} \\
  \matz \end{smallmatrix}]\) and, since \(\mat{U}\) is nonsingular,
  \(\mat{U}\sys_1 = [\begin{smallmatrix} \mat{V} \\ \matz \end{smallmatrix}]\)
  has the same rank profile as \(\sys_1\), which is \(\rps_1\). In
  particular, \(\mat{V}\) has full row rank, and then the triangular form of
  \(\mat{U}\sys\) implies that its rank profile is the concatenation of
  \(\rps_1\) and of \(\rps_2\), the latter being the rank profile of \(\sys_2\)
  shifted by adding \(\lfloor \cdim/2 \rfloor\) to all entries. Since
  \(\mat{U}\sys\) and \(\sys\) have the same rank profile,
  \(\sys\) has rank profile \((\rps_1,\rps_2)\).

\textbf{Case \(\cdim \le \rdim/2\).} The algorithm runs
\crefrange{step:algo:kbas_rkprof:rbas_params:start}{step:algo:kbas_rkprof:rbas_return}
and returns.

The basis \(\abas\) of \(\modRel{\mmat}{\sys}\) at
\cref{step:algo:kbas_rkprof:rbas:end} is such that \(\abas \sys = \mat{Q} \mmat\)
for some \(\mat{Q} \in \pmatRing{\rdim}{\cdim}\).  Since both \(\abas\) and
\(\mmat\) are nonsingular, with  \(\mmat\) being upper triangular, implies that
\(\mat{Q} = \abas\sys\mmat^{-1}\) has the same rank profile as
\(\sys\).  By the construction at
\crefrange{step:algo:kbas_rkprof:residual:start}{step:algo:kbas_rkprof:residual:end}, we see that there is an 
\(\rdim\times\rdim\) permutation matrix \(\mat{P}\) such that
\[
  \mat{P} \mat{Q}
  =
  \mat{P} \abas \sys \mmat^{-1}
  = 
  \begin{bmatrix}
    \matrows{\abas}{I^c} \sys \mmat^{-1} \\
    \matrows{\abas}{I} \sys \mmat^{-1}
  \end{bmatrix}
  =
  \begin{bmatrix}
    \mat{G} \\
    \matz
  \end{bmatrix}.
\]
Thus \(\mat{G}\) has the same rank profile as \(\mat{P}\mat{Q}\), and hence 
 the same rank profile as \(\sys\). Thus, to conclude the proof for the
rank profile, it suffices to verify that the recursive call at
\cref{step:algo:kbas_rkprof:reccall:reccall} computes the rank profile of \(\mat{G}\),
which is true provided that \(\shiftt\) satisfies the requirements \(\shiftt
\ge \tuple{0}\) and \(\shiftt \ge \rdeg{\mat{G}}\). We prove this in the next
paragraph.

Observe that the shift built at \cref{step:algo:kbas_rkprof:rbas_params:end}
satisfies \(\shiftr \ge \rdeg{\sys}\). This implies
\(\rdeg[\shiftr]{\matrows{\abas}{I^c}} \ge \rdeg{\matrows{\abas}{I^c} \sys}\) and,
defining \(\cdegs = \cdeg{\mmat}\),
\begin{align*}
  \shiftt \ge \rdeg{\matrows{\abas}{I^c} \sys} & - (\gamma,\ldots,\gamma) = \rdeg{\mat{G}\mmat} - (\gamma,\ldots,\gamma) \\
                                                                      & = \rdeg[(-\gamma,\ldots,-\gamma)]{\mat{G}\mmat}
                                                                       \ge \rdeg[-\cdegs]{\mat{G}\mmat},
\end{align*}
where the last inequality comes from \((\gamma,\ldots,\gamma) \le \cdegs\). Now
the fact that \(\mmat\) is in Hermite form ensures that it is in
\(-\cdegs\)-reduced form with \(\rdeg[-\cdegs]{\mmat} = \shiftz\), so that the
predictable degree property yields \(\rdeg[-\cdegs]{\mat{G}\mmat} =
\rdeg{\mat{G}}\). Thus \(\shiftt \ge \rdeg{\mat{G}}\), and \(\shiftt
\ge \shiftz\) follows since \(\mat{G}\) has no zero row by construction.

This also ensures that \(\kbas_2\) is a \(\shiftt\)-weak Popov basis of
\(\modKer{\mat{G}} = \modKer{\matrows{\abas}{I^c} \sys}\). 
\cref{lem:predictable_pivot,lem:weak_popov_multiplication} then imply that \(\kbas_2
\matrows{\abas}{I^c}\) is \(\shiftr\)-weak Popov, with \(\shiftr\)-pivot index a
subset of that of \(\matrows{\abas}{I^c}\). Since the latter is disjoint from the
\(\shiftr\)-pivot index of \(\matrows{\abas}{I}\) and since \(\kbas_2
\matrows{\abas}{I^c}\) and \(\matrows{\abas}{I}\) are both \(\shiftr\)-weak
Popov, it follows that \(\kbas\) is \(\shiftr\)-weak Popov. Since
\(\shifts\) and \(\shiftr\) only differ by a constant, \(\kbas\) is
\(\shifts\)-weak Popov.

By construction, \(\matrows{\abas}{I} \sys = \matz\) and \(\matz = \kbas_2
\mat{G} = \kbas_2 \matrows{\abas}{I^c} \sys\), and so  \(\kbas \sys = \matz\). It
remains to prove that any \(\row{p} \in \modKer{\sys}\) is a
\(\polRing\)-linear combination of the rows of \(\kbas\). Since \(\row{p}\) is
in \(\modRel{\mmat}{\sys}\), we get \(\row{p} = \row{q} \abas =
\row{q}_I \matrows{\abas}{I} + \row{q}_I^c \matrows{\abas}{I^c}\) for some
\(\row{q} \in \pmatRing{1}{\rdim}\) and its subvectors \(\row{q}_I\) and
\(\row{q}_I^c\) with indices in \(I\) and \(I^c\), respectively. Then,
\[
  \matz = \row{p} \sys
  = \row{q}_I \matrows{\abas}{I} \sys + \row{q}_I^c \matrows{\abas}{I^c}\sys
  = \row{q}_I^c \matrows{\abas}{I^c}\sys,
\]
which gives
\(\row{q}_I^c \in \modKer{\matrows{\abas}{I^c}\sys} = \modKer{\mat{G}\mmat} =
\modKer{\mat{G}}\). Therefore \(\row{q}_I^c = \row{r} \kbas_2\) for some vector
\(\row{r}\), and we get \(\row{p} = \row{q}_I \matrows{\abas}{I} +
\row{r}\kbas_2 \matrows{\abas}{I^c}\).

\emph{Remark:} As one can see above, \(\mmat\) is required to be in Hermite
normal form only for ensuring that the rank profile is not modified when
right-multiplying by \(\mmat\). Hence, if one is only interested in a kernel
basis, any column reduced matrix \(\mmat\) will do.

\subsection{Proof of complexity}
\label{sec:kbas_rkprof:complexity}

The efficiency is based on three main ingredients. First, a fast algorithm for
computing an \(\shifts\)-weak Popov basis of \(\modRel{\mmat}{\sys}\). Second,
the fast multiplication of matrices which have unbalanced, but controlled,
shifted row degrees. Third, the next lemma, which is a generalization and
variant of \cite[Thm.\,3.6]{ZhLaSt12}: it states that the relation basis at
\cref{step:algo:kbas_rkprof:rbas:end} yields a substantial amount of kernel
rows, effectively reducing the number of rows that remain to be found.

\begin{lemma}
  \label{lem:partial_kernel_via_approx:ker}
  Let $\sys \in \pmatRing{\rdim}{\cdim}$ have rank \(r\), and let $\shifts\in
  \NN^\rdim$ such that \(\shifts \ge \rdeg{\sys}\). Let \(\kbas \in
  \pmatRing{(\rdim-\rk)}{\rdim}\) be an \(\shifts\)-reduced basis of
  \(\modKer{\sys}\). For any \(k>0\), at most \(\lfloor k \rfloor\) rows of
  \(\kbas\) have \(\shifts\)-degree more than or equal to \(\order = \lceil
  \sumTuple{\shifts}/k \rceil\). Then, let \(\mmat \in
  \pmatRing{\cdim}{\cdim}\) be column reduced with \(\min(\cdeg{\mmat}) \ge
  \order\). For any \(\shifts\)-reduced basis \(\abas \in
  \pmatRing{\rdim}{\rdim}\) of \(\modRel{\mmat}{\sys}\), at most \(\rk+\lfloor
  k \rfloor\) rows of \(\abas\) are not in \(\modKer{\sys}\).
\end{lemma}
\begin{proof}
  If \(\rho\) is the number of rows of \(\kbas\) whose \(\shifts\)-degree is
  \(\ge \order\), then \(\sumTuple{\rdeg[\shifts]{\kbas}} \ge \rho \order\).
  Thus, from the bound \(\sumTuple{\rdeg[\shifts]{\kbas}} \le
  \sumTuple{\shifts} \le k\order\) (see \cref{thm:degree_rank:rdeg} of
  \cref{thm:degree_rank}), we get \(\rho \le k\). It follows that there are
  \(\sigma = \rdim-\rk-\rho \ge \rdim-\rk-k\) rows of \(\kbas\) whose
  \(\shifts\)-degree is \(<\order\).

  We claim that, as a consequence, there are at least \(\sigma\) rows of
  \(\abas\) which are in \(\modKer{\sys}\). First, since all rows of \(\kbas\)
  are also in \(\modRel{\mmat}{\sys}\), by minimality of the \(\shifts\)-degree
  of \(\abas\), there are at least \(\sigma\) rows of \(\abas\) which have
  \(\shifts\)-degree \(<\order\). The claim then follows from the fact that any
  \(\row{p} \in \modRel{\mmat}{\sys}\) such that \(\rdeg[\shifts]{\row{p}} <
  \order\) is in \(\modKer{\sys}\). Indeed, \(\row{p} \sys = \row{q} \mmat\)
  for some \(\row{q} \in \pmatRing{1}{\cdim}\). On the one hand,
  \(\rdeg{\row{p}\sys} \le \rdeg[\shifts]{\row{p}} < \order\), since
  \(\shifts\ge\rdeg{\sys}\). On the other hand, \(\mmat\) is column reduced
  with \(\cdegs=\cdeg{\mmat}\), and thus it is also \(-\cdegs\)-reduced with
  \(\rdeg[-\cdegs]{\mmat} = \shiftz\). Hence, assuming \(\row{p}\sys\neq\matz\) is
  nonzero (and therefore \(\row{q}\neq\matz\)), and using \(\order \le
  \min(\cdegs)\) as well as the predictable degree property, we obtain
  \begin{align*}
    \order > \rdeg{\row{p}\sys} = \rdeg{\row{q} \mmat} & = \order + \rdeg[(-\order,\ldots,-\order)]{\row{q} \mmat} \\
                                              & \ge \order + \rdeg[-\cdegs]{\row{q} \mmat} 
                                               = \order + \rdeg{\row{q}} \ge \order.
  \end{align*}
  This is a contradiction, hence \(\row{p}\sys = \matz\), i.e.~\(\row{p} \in
  \modKer{\sys}\).

  In conclusion, \(\abas\) has at most \(\rdim - \sigma \le \rk+k\) rows not in
  \(\modKer{\sys}\).
\end{proof}

Note that if we take  \(k\le \frac{\rdim-\rk}{2}\), then this number is \(\rk+k \le
\frac{\rdim+\rk}{2}\). For example, if \(\cdim<\rdim\) and no information on
\(\rk\) is known, one can take \(k=\frac{\rdim-\cdim}{2}\), and then \(\rk+k =
\frac{\rdim+\cdim}{2}\). This is the choice made in \cref{algo:kbas_rkprof},
which leads to \(\order=\lceil 2\sumTuple{\shifts}/(\rdim-\cdim) \rceil\).
Furthermore, since \(2\cdim\le \rdim\) in that algorithm, we obtain \(\rk+k =
\frac{\rdim+\cdim}{2} \le \frac{3\rdim}{4}\).

\begin{corollary}
  \label{cor:size_rbas_residual}
  At \cref{step:algo:kbas_rkprof:residual:product} of \cref{algo:kbas_rkprof},
  the matrix \(\mat{G}\) has at most \(\frac{3\rdim}{4}\) rows.
\end{corollary}

This lemma ensures that, when the algorithm enters
\crefrange{step:algo:kbas_rkprof:rbas_params:start}{step:algo:kbas_rkprof:rbas_return},
then the number of rows becomes at most \(3\rdim/4\) in the recursive call (and
the number of columns is unchanged). On the other hand, when the algorithm
enters
\crefrange{step:algo:kbas_rkprof:split_first}{step:algo:kbas_rkprof:split_return},
then the number of columns becomes at most \(\lceil \cdim/2 \rceil\) in each of
the two recursive calls (and the number of rows remains bounded from above by
\(\rdim\)).

Note that we have proved in \cref{sec:kbas_rkprof:correctness} that in each
recursive call, the input shift is an upper bound on the row degrees of the
input matrix. Now, we observe further that each of these shifts has a sum of
entries at most \(\sumTuple{\shifts}\). This is clear at
\cref{step:algo:kbas_rkprof:split_rec_one} which uses the input shift
\(\shifts\), and at \cref{step:algo:kbas_rkprof:split_rec_two} since the shift
satisfies \(\sumTuple{\rdeg[\shifts]{\kbas_1}} \le \sumTuple{\shifts}\)
according to \cref{thm:degree_rank:rdeg} of \cref{thm:degree_rank}.  Now, at
\cref{step:algo:kbas_rkprof:reccall:shift,step:algo:kbas_rkprof:reccall:reccall},
the shift \(\shiftt\) has entries at most those of the subtuple
\(\shifts_{I^c}\), since
\[
  \rdeg[\shiftr]{\abas} - (\gamma,\ldots,\gamma)
  %\le \rdeg[\shifts]{\mat{A}} - (\gamma,\ldots,\gamma)
  = \shifts + \pivDegs - (\gamma+\mu,\ldots,\gamma+\mu)
\]
where \(\pivDegs\) is the \(\shifts\)-pivot degree of \(\abas\), with
\(\pivDegs \le (\gamma,\ldots,\gamma)\) under our assumption \(\cdeg{\mmat} =
(\gamma,\ldots,\gamma)\).

Recall the notation \(\dd = \sumTuple{\shifts}\), and note that \(\rdim \lceil
\dd/\rdim \rceil = \Theta(\rdim + \dd)\).

Based on \cite[Thm.\,3.7]{ZhLaSt12}, one can verify that the matrix products at
\cref{step:algo:kbas_rkprof:split_residual,step:algo:kbas_rkprof:split_return,step:algo:kbas_rkprof:residual:product,step:algo:kbas_rkprof:reccall:product},
use \(\softO{\rdim^{\expmm-2} (\rdim+\cdim) (\rdim+\dd)}\) operations in
\(\field\). Note that the right-multiplication by \(\mmat^{-1}\) at
\cref{step:algo:kbas_rkprof:residual:product} is only a matter of univariate
polynomial exact division: the matrix \(\matrows{\abas}{I^c} \, \sys\) is known
to be a left multiple of \(\mmat\) by construction, and the latter matrix is
diagonal by assumption.

It remains to observe that the computation of \(\abas\) at
\cref{step:algo:kbas_rkprof:rbas:end} costs \(\softO{\rdim^{\expmm-1} \cdim
\order}\) operations in \(\field\), since \(\mmat\) is diagonal
\cite[Thm.\,1.4]{Neiger16}.
Since we are in the case \(\cdim \le \rdim/2\), we have \(\rdim-\cdim \ge
\rdim/2\) and thus
\[
  \order =
  \left\lceil \frac{2 \sumTuple{\shiftr}}{\rdim-\cdim} \right\rceil
  \in \bigO{1 + \frac{\dd}{\rdim-\cdim}}
  \subseteq \bigO{1+\frac{\dd}{\rdim}}.
\]
Hence the above cost for computing \(\abas\) is in \(\softO{\rdim^{\expmm-2}
\cdim (\rdim+\dd)}\).

We conclude that all computations apart from recursive calls use a total of
\(\softO{\rdim^{\expmm-2} (\rdim+\cdim) (\rdim+\dd)}\) operations in
\(\field\), leading to the cost bound announced in \cref{thm:algo:kbas_rkprof}.

\section{Finding the column rank profile and linearly independent rows}
\label{sec:rkprof_indeprows}

Let $\sys$ be a polynomial matrix of rank $\rk$. This section presents a
rank-sensitive algorithm to find both the rank profile of \(\sys\) and a set of
\(\rk\) rows of \(\sys\) which are \(\polRing\)-linearly independent. In
particular, this information locates an \(\rk\times\rk\) nonsingular submatrix
of $\sys$.

\subsection{Algorithm}
\label{sec:rkprof_indeprows:algo}

The idea is to maintain a subset \(U\) of the top rows of \(\sys\), which are
known to have full rank, and to incorporate new rows from the bottom part of
\(\sys\). Precisely, \(U\) locates \(k\) rows with index in
\((1,\ldots,\stind-1)\), and the next step finds a set of rows of maximal rank
in the matrix \(\mat{G}\) formed by joining these \(k\) rows
\(\matrow{\sys}{U}\) with the \(k\) rows with indices
\((\stind,\ldots,\stind+k-1)\) of \(\sys\) (or only up to \(\rdim\) if
\(\stind+k-1 \ge m\)).

Finding a set of rows of maximal rank of \(\mat{G}\) is done efficiently via
\cref{algo:kbas_rkprof} and the property in \cref{thm:degree_rank:rank} of
\cref{thm:degree_rank}, which locates independent rows from the
\(\shifts\)-pivot index of the kernel basis. Since the call to
\cref{algo:kbas_rkprof} also provides the column rank profile of \(\mat{G}\),
we eventually obtain \(I\) and \(J\) identifying a nonsingular submatrix of
\(\sys\) of size \(\rk \times \rk\), with \(J\) the rank profile of
\(\matrow{\sys}{I}\). By \cref{lem:ker_crp_equal}, the latter is also the rank
profile of \(\sys\).

Starting with \(k=1\) and \(U\) locating the first nonzero row of \(\sys\),
this leads to a rank-sensitive algorithm, which at any stage considers a
submatrix of \(\sys\) with \(\cdim\) columns and at most \(2k \le 2\rk\) rows.

\begin{algorithm}[t]
  \caption{\textsc{ColumnRankProfile}\((\sys, \stind, U)\)}
  \label{algo:column_rank_profile}
  \begin{algorithmic}[1]
    \Require{a matrix \(\sys \in \pmatRing{\rdim}{\cdim}\),
      an integer \(\stind \in \{1,\ldots,\rdim+1\}\),
      a list $U \subseteq \{1,\ldots,\stind-1\}$ of size \(k\ge 0\)}
    \Assume{\(k=0\) or \(\rank{\matrow{\sys}{U}} = \rank{\matrow{\sys}{1..\stind-1}} = k\)}
    \Ensure{lists \(I \subseteq \{1,\ldots,\rdim\}\) and \(J \subseteq
      \{1,\ldots,\cdim\}\), both of size \(\rk = \rank{\sys}\), such that
      \(\matsub{\sys}{I}{J} \in \pmatRing{\rk}{\rk}\) is nonsingular and \(J\)
      is the rank profile of \(\sys\)}
   
    \State\InlineIf{\(k=\cdim\)}{\Return \(U,(1,\ldots,\cdim)\)}
        \label{step:algo:column_rank_profile:k_is_n}

    \If{\(k=0\)} \label{step:algo:column_rank_profile:k_is_0}
      \State \(i \gets\) index of the first nonzero row of \(\sys\)
          \label{step:algo:column_rank_profile:k_is_0_find_row}
      \State \Return \(\Call{ColumnRankProfile}{\sys,i+1,(i)}\)
          \label{step:algo:column_rank_profile:k_is_0_return}
    \EndIf

    \State\CommentLine{\(k>0\) independent rows are known among the rows \(1,\ldots,\stind-1\)}

    \State \(\ell \gets \min(k, \rdim-\stind+1)\)
        \Comment{now incorporate rows \(\stind,\ldots,\stind+\ell-1\)}
          \label{step:algo:column_rank_profile:choose_ell}

    \State \(V \gets U \cup (\stind,\stind+1,\ldots,\stind+\ell-1)\)
          \label{step:algo:column_rank_profile:build_V}

    \State \(\mat{G} \gets \matrow{\sys}{V} \in \pmatRing{(k+\ell)}{\cdim}\); \(\shifts \gets \rdeg{\mat{G}}\)
          \label{step:algo:column_rank_profile:build_mat_shift}

    \State \(\kbas, \rps \gets \Call{KernelBasis-RankProfile}{\mat{G}, \shifts}\)
          \label{step:algo:column_rank_profile:kbas_rkprof}

    \State \((\pivInds,\pivDegs)\) \(\gets\) \(\shifts\)-pivot
          profile of \(\kbas\); \(\nonPivInds \gets \{1,\ldots,k+\ell\}\setminus\pivInds\)
          \label{step:algo:column_rank_profile:nonpivinds}

    \State \(U' \gets \emptyset\)
          \label{step:algo:column_rank_profile:update_start}
    \For{\(i \in \nonPivInds\)}
      \State\InlineIf{\(i\le k\)}{add the \(i\)th element of \(U\) to \(U'\)}
      \State\InlineElse{add \(\stind+i-k-1\) to \(U'\)}
          \label{step:algo:column_rank_profile:update_end}
    \EndFor

    \State \(\stind' \gets \stind+\ell\)
          \label{step:algo:column_rank_profile:update_stind}

    \State\InlineIf{\(\stind'>\rdim\)}{\Return \(U',\rps\)}
          \label{step:algo:column_rank_profile:exit_if_done}

    \State\Return \(\Call{ColumnRankProfile}{\sys, \stind', U'}\)
          \label{step:algo:column_rank_profile:rec_call}
  \end{algorithmic}
\end{algorithm}

\begin{theorem}
  \label{thm:algo:column_rank_profile}
  Let \(\sys \in \pmatRing{\rdim}{\cdim}\) have rank \(\rk\).  Assume that
  \(\rdeg{\sys}\) is nondecreasing, and that \((\stind,U)\) satisfies the input
  requirements. Then \(\Call{ColumnRankProfile}{\sys,\stind,U}\) uses
  \({\softO{\rk^{\omega-2} \cdim (\rdim+D)}}\) operations in \(\field\), where
  \(\dd\) is the sum of the nonnegative entries of \(\rdeg{\sys}\).  It returns
  lists \(I \subseteq \{1,\ldots,\rdim\}\) and \(J \subseteq
  \{1,\ldots,\cdim\}\), both of size \(\rk\), such that \(\matsub{\sys}{I}{J}
  \in \pmatRing{\rk}{\rk}\) is nonsingular and \(J\) is the rank profile of
  \(\sys\).
\end{theorem}

Before proving the theorem we note that, if nothing particular is known about
\(\sys\) a priori, one can call this algorithm with \(\stind=1\) and \(k=0\)
(meaning \(U=\emptyset\)). One can also permute the rows of \(\sys\) to ensure
that its row degrees are nondecreasing.

\begin{corollary}
  \label{cor:algo:column_rank_profile}
  Given \(\sys \in \pmatRing{\rdim}{\cdim}\), one can locate an \(\rk \times
  \rk\) nonsingular submatrix \(\matsub{\sys}{I}{J}\) of \(\sys\) using
  \({\softO{\rk^{\omega-2} \cdim (\rdim+D)}}\) operations in \(\field\), where
  \(\rk\) is the rank of \(\sys\), \(J\) is the column rank profile of
  \(\matrow{\sys}{I}\), and \(\dd\) is the sum of the degrees of the nonzero
  rows of \(\sys\).
\end{corollary}

\subsection{Proof of correctness}
\label{sec:rkprof_indeprows:correctness}

If \(k=\cdim\), the requirement \(\rank{\matrow{\sys}{U}} = \cdim\) implies
that the \(\cdim\times\cdim\) matrix \(\matsub{\sys}{U}{1..n}\) is nonsingular,
proving the correctness of \cref{step:algo:column_rank_profile:k_is_n}.

If \(k=0\), then the correctness
\crefrange{step:algo:column_rank_profile:k_is_0}{step:algo:column_rank_profile:k_is_0_return}
follows from the fact the input requirements are satisfied for \(\stind=i+1\)
and \(U = (i)\).

The integer \(\ell\) at \cref{step:algo:column_rank_profile:choose_ell} is such
that \(0 \le \ell \le k\) and \(\stind+\ell-1 \le \rdim\). Then, at
\cref{step:algo:column_rank_profile:build_V}, the list \(V\) contains
\(k+\ell\) distinct indices, with the first \(k\) in \(U\) and the others in
\(\{\stind,\ldots,\rdim\}\). As a result, the matrix \(\mat{G} =
\matrow{\sys}{V}\) at \cref{step:algo:column_rank_profile:build_mat_shift} has
rank between \(k\) and \(k+\ell\). By \cref{thm:degree_rank:rank} of
\cref{thm:degree_rank}, the list \(\nonPivInds\) computed at
\cref{step:algo:column_rank_profile:nonpivinds,step:algo:column_rank_profile:kbas_rkprof}
identifies \(\rank{\mat{G}}\) rows of \(\mat{G}\) which are
\(\polRing\)-linearly independent.

These rows provide a set of rows of \(\sys\) which have maximal rank among its
first \(\stind+\ell-1\) rows. The role of
\crefrange{step:algo:column_rank_profile:update_start}{step:algo:column_rank_profile:update_end}
is simply to link the row indices in \(\mat{G}\), as they appear in
\(\nonPivInds\), to the corresponding row indices in \(\sys\). This therefore
provides \(U'\) such that $U' \subseteq \{1,\ldots,\stind+\ell-1\}$, whose size
\(k'\) is between \(k\) and \(k+\ell\), and such that \(\matrow{\sys}{U'}\) has
full row rank with \(\rank{\matrow{\sys}{U'}} =
\rank{\matrow{\sys}{1..\stind+\ell-1}}\).

Then \cref{step:algo:column_rank_profile:update_stind} updates \(\stind\) to
\(\stind'\), to reflect that we have now covered all rows with indices in
\(\{1,\ldots,\stind+\ell-1\}\), and that we can proceed with the remaining rows
starting at index \(\stind' = \stind+\ell\).

In the case where \(\stind' = \stind+\ell>\rdim\)
(\cref{step:algo:column_rank_profile:exit_if_done}), all rows of \(\sys\) have
been processed and the algorithm can return \(U'\) and \(J\).  The correctness
concerning \(U'\) has been discussed above, and the fact that \(J\) is the rank
profile of \(\matrow{\sys}{U'}\) follows from \cref{thm:algo:kbas_rkprof}.

Otherwise, if \(\stind'\le\rdim\), we proceed with the remaining bottom part of
\(\sys\) recursively (\cref{step:algo:column_rank_profile:rec_call}). The
properties of \(U'\) described above show that the input requirement are
satisfied in this recursive call.

Thus, for correctness, it only remains to observe the algorithm terminates
since all recursive calls involve a strictly larger set \(U\) (whose size is
bounded by the rank \(\rk\) of \(\sys\)), or a strictly larger index
\(\stind\) (which is bounded by \(\rdim+1\)).

\emph{Remark:} the independent rows of \(\sys\) found via the \(\shifts\)-pivot
profile of \(\kbas\) do not necessarily contain the \(k\) independent rows that
were already identified by \(U\); in other words, \(U\) is not necessarily
contained in \(U'\). Having \(U\subseteq U'\) would have guaranteed that \(I\)
is the row rank profile of \(\sys\), yet the straightforward modification of
this algorithm that would ensure \(U\subseteq U'\) involves a different choice
of shift \(\shifts\) which would make
\cref{step:algo:column_rank_profile:kbas_rkprof} too costly.

\subsection{Proof of complexity}
\label{sec:rkprof_indeprows:complexity}

The only costly operation performed in \cref{algo:column_rank_profile}, apart
from recursive calls, is the computation of the kernel basis and rank profile
of the matrix \(\mat{G}\), via \cref{algo:kbas_rkprof} at
\cref{step:algo:column_rank_profile:kbas_rkprof}. The main task for the
complexity analysis is therefore to analyze how the dimensions and row degrees
of \(\mat{G}\) evolve during the run of the algorithm.

Fix some input \((\sys,\stind,U)\), and assume the cardinality \(k\) of \(U\)
is nonzero; otherwise, we are brought to this situation by
\cref{step:algo:column_rank_profile:k_is_0_return} after at most \(mn\) zero
tests performed by \cref{step:algo:column_rank_profile:k_is_0_find_row}. For
%% note: this assumption also ensures F only has nonzero rows, and rdeg(F) >= 0
this input, let \(\rho\) be the number of recursive steps before arriving at a
base case (either \cref{step:algo:column_rank_profile:k_is_n} or
\cref{step:algo:column_rank_profile:exit_if_done}). Let \((\stind_0,U_0) =
(\stind,U)\) be the original input and \((\stind_1,U_1), \ldots,
(\stind_\rho,U_\rho)\) be the input of the successive recursive calls when
running the algorithm on \((\sys,\stind,U)\).  Let also \(k_i\) be the
cardinality of \(U_i\) for \(1 \le i \le \rho\). Observe that \(\stind_0 <
\stind_1 < \cdots < \stind_\rho\) and \(k_0 \le k_1 \le \cdots \le k_\rho\),
but recall from the remark in \cref{sec:rkprof_indeprows:correctness} that
\(U_i\) is not necessarily a subset of \(U_{i+1}\).

Let \(\ell_i = \min(k_i , \rdim-\stind_i+1)\) as in
\cref{step:algo:column_rank_profile:choose_ell}, and \(\mat{G}_i\) be the
matrix built at \cref{step:algo:column_rank_profile:build_mat_shift}, which has
\(k_i+\ell_i\) rows. By \cref{thm:algo:kbas_rkprof},
\cref{step:algo:column_rank_profile:kbas_rkprof} costs
\[
  \softO{(k_i+\ell_i)^{\expmm-2} (k_i+\ell_i+\cdim)(k_i+\ell_i+\dd_i)},
\]
where \(\dd_i = \sumTuple{\rdeg{\mat{G}_i}}\). Now, since \(k_i =
\rank{\matrow{\sys}{U_i}} \le \rk \le \cdim\) and \(k_i+\ell_i \le 2k_i\), the
above cost bound is within \(\softO{\rk^{\expmm-2} \cdim(k_i+\dd_i)}\).  The
rest of this section shows that \(\sum_{0 \le i \le \rho} k_i+\dd_i\) is in
\(\bigO{\rdim+\dd\log(\rk)}\), which proves the complexity bound
\(\softO{\rk^{\expmm-2} \cdim(\rdim+\dd)}\).

Due to \cref{step:algo:column_rank_profile:update_stind}, \(\stind_{i+1} =
\stind_i + \ell_i\) for \(0\le i < \rho\). Observe that \(\ell_i = k_i\) for
\(i<\rho\); otherwise \(\ell_i = \rdim - \stind_i + 1\) and \(\stind_{i+1} =
\theta_i + \ell_i = \rdim+1\), hence the algorithm would stop at
\cref{step:algo:column_rank_profile:exit_if_done} before the \(\rho\)-th
recursive call. It follows that \(\stind_i = \stind_0 + \sum_{0\le j<i} k_j\)
for \(0 \le i \le \rho\). \Cref{step:algo:column_rank_profile:exit_if_done}
ensures that \(\stind_\rho \le \rdim\) in the last recursive call, hence \(k_0 +
\cdots + k_{\rho-1}  = \stind_\rho - \stind_0 \le \rdim\). We finally deduce
\(k_0 + \cdots + k_\rho \le \rdim + k_\rho \le \rdim + \rk \le 2\rdim\).

It remains to prove \(\dd_0+\cdots+\dd_\rho \in \bigO{\dd\log(\rk)}\). By
construction, \(\mat{G}_i\) consists of \(k_i+\ell_i\) rows among the first
\(\{1,\ldots,\stind_i+\ell_i-1\}\) rows of \(\sys\). Consequently, since
\(\rdeg{\sys}\) is nondecreasing, \(\dd_i \le
\sumTuple{\rdeg{\matrow{\sys}{S_i}}}\) where \(S_i = \{\stind_{i} - k_i,
\stind_i - k_i + 1, \ldots, \stind_{i} + \ell_i -1\}\). We claim that a given
row \(j\in \{1,\ldots,\rdim\}\) of \(\sys\) may appear in at most \(\lfloor
\log_2(\rk) \rfloor + 2\) sets among \(S_0,\ldots,S_\rho\), that is,
\(\card{\{0 \le i \le \rho \mid j \in S_i\}} \le \lfloor \log_2(\rk) \rfloor +
2\). Below we prove this claim, which concludes the proof since then
\begin{align*}
  \textstyle\sum\limits_{0\le i\le \rho} D_i & \le \textstyle\sum\limits_{0\le i\le \rho} \sumTuple{\rdeg{\matrow{\sys}{S_i}}}
                = \textstyle\sum\limits_{0\le i\le \rho} \textstyle\sum\limits_{j \in S_i} \rdeg{\matrow{\sys}{j}} \\
                            & = \textstyle\sum\limits_{1 \le j \le \rdim} \, \sum\limits_{\substack{0 \le i \le \rho \\ j \in S_i}} \rdeg{\matrow{\sys}{j}}
  %% & \quad\;\; = \sum_{1 \le j \le \rdim} \rdeg{\matrow{\sys}{j}} \, \card{\{0 \le i \le \rho \mid j \in S_i\}} \\
                            %% & \le \sum_{1 \le j \le \rdim} \rdeg{\matrow{\sys}{j}}\log_2(\rk) 
    \le \dd (\lfloor \log_2(\rk) \rfloor+2) .
\end{align*}

Since \((\min(S_i))_i\) and \((\max(S_i))_i\) are nondecreasing, having \(j\in
S_{i_1} \cap S_{i_2}\) for some \(i_1 < i_2\) implies \(j \in S_i\) for all
\(i_1 \le i \le i_2\). So we consider \(j \in S_i,\ldots,S_{i+c-1}\) for some
\(c>0\) and \(0\le i \le \rho-c+1\).

The fact that \(j\in S_i\) implies \(j \le \stind_{i}+\ell_i - 1 = \stind_{i+1}
- 1\). On the other hand, \(j\in S_{i+\gamma}\) for \(0 \le \gamma \le c-1\)
implies \(j \ge \stind_{i+\gamma} - k_{i+\gamma}\). We deduce \(k_{i+\gamma}
\ge \stind_{i+\gamma} - \stind_{i+1} + 1 = k_{i+\gamma-1} + \cdots + k_{i+2} +
k_{i+1} + 1\). Starting from \(k_{i+1} \ge 1\), using this inequality
iteratively for \(\gamma=2,\ldots,c-1\) shows that \(k_{i+c-1} \ge 2^{c-2}\).
Since \(k_{i+c-1} \le \rk\), we get \(c \le \lfloor \log_2(\rk) \rfloor + 2\).

\section{Topics for further research}

Our algorithm can find the rank profile with a rank-sensitive cost. However the
same cannot be said for such computations as kernel basis, column basis, and
approximant/order bases; or for computing normal forms such as Hermite or Popov.
We would like to make progress on filling this gap thanks to the new results in
this paper. In addition, we are interested in applying our work in situations
(including those mentioned in the introduction) where rank-sensitive algorithms
would allow one to tackle significantly larger problems.

Another feature of our algorithms is that the complexity depends on the average
row degree of the input matrix. However, they do not handle unbalanced column
degrees, where the matrix might have average row degree close to the global
degree but average column degree quite smaller. We would like to improve the
cost towards the minimum of the average of both row and column degrees, or even
on a notion of generic determinant bound \cite[Sec.\,6]{GuSaStVa12} generalized
to rectangular matrices. We believe that \emph{partial linearization}
\cite[Sec.\,3]{Storjohann06} \cite[Sec.\,6]{GuSaStVa12} may lead to such an
improvement.

Finally, while it has been popular in recent times to give \(\softO{\cdot}\)
complexities which hide log terms, there remains a strong interest in the more
precise \(\bigO{\cdot}\) measure. In fact this is often the first audience
question asked when an algorithm is presented with \(\softO{\cdot}\)
complexity. We would like to determine the logarithmic terms for the algorithms
presented in this paper. Although technical, this seems feasible since the
logarithmic factors in the complexity of the core tools are now well
understood, specifically approximant bases
\cite{GiJeVi03,ZhoLab12,JeannerodNeigerVillard2020} and multiplication with
unbalanced degrees \cite{ZhLaSt12,JeNeScVi17}.

\begin{figure*}[t]
  \caption{SageMath code for \cref{example:main,example:cont_main,example:continued}.
    \textmd{\emph{This code is written using SageMath (version 9.3 or later required)
      and illustrates many of the points in the three listed examples: running
      this code will show the matrices and some additional information. This
      code can also be easily adapted to make related experiments.}}}
  \label{fig:sage_code}
  \centering
  \begin{minted}[frame=single]{python}
  # For a detailed documentation of functionalities for univariate polynomial matrices, including
  # the minimal_kernel_basis and minimal_approximant_basis methods used below, see
  # https://doc.sagemath.org/html/en/reference/matrices/sage/matrix/matrix_polynomial_dense.html
  # (the code below requires SageMath >=9.3; some other functionalities require SageMath 9.4 or 9.5)

  pR.<x> = GF(2)[]
  F = Matrix(pR, 5, 5, \
          [[x^2,x^3+1,x^8+x^6+x^4+x^3+x^2+x,x^4+1,x^3+1], \
          [0,x^4+1,x^5+x^4+x^3+x^2,x+1,x^2+1], \
          [0,x^2+1,x+1,0,1], \
          [0,0,x^8+1,x^4+1,0], \
          [0,0,x^4+1,1,0]])

  print(f"Input matrix F:\n{F}\n")

  K0 = F.minimal_kernel_basis()
  print(f"Minimal kernel basis with shift s=0\n{K0}")
  piv = [pi+1 for pi in K0.leading_positions()]
  print(f"Its pivot indices are {piv}\n")

  s = F.row_degrees()
  K = F.minimal_kernel_basis(shifts=s)
  print(f"Minimal kernel basis K with shift s=rdeg(F)\n{K}")
  piv_s = [pi+1 for pi in K.leading_positions(shifts=s)]
  print(f"Its s-pivot indices are {piv_s}\n")

  F1 = F[:,:2] ; tau = 18
  A = F1.minimal_approximant_basis(tau,shifts=s)
  print(f"Approximant basis of first 2 columns at order {tau}:\n{A}")
  # set I at Line 14 of Algorithm 1:
  I = [i for i in range(5) if A[i,:].row_degrees(shifts=s)[0] < tau]
  # --> gives 3 indices I == [1,3,4] so we have the whole kernel basis
  # since here we do know rank(F[:,:2]) = 2 hence kernel rank 5-2==3
  print(f"--> indices of rows in kernel: {[i+1 for i in I]}\n")
  K1 = A[I,:]; t = K1.row_degrees(shifts=s)

  F2 = K1 * F[:,2:]
  print(f"Residual matrix F2 for second call:\n{F2}")
  K2 = F2[:,0].minimal_kernel_basis(shifts=t)

  print(f"Kernel basis K1' of first column of F2:\n{K2}\n")

  print(f"Test K1' * F2 == 0 --> {K2*F2 == 0}, so in fact K2 = K1'")

  print(f"Verify K2*K1 is the above s-weak Popov matrix K --> {K2*K1 == K}")
  \end{minted}
\end{figure*}

% Fakesection --> fold acknowledgement / biblio (for vimtex plugin, do not remove this line, thanks)
\bibliographystyle{ACM-Reference-Format}
%%\bibliography{biblio}


\end{document}